\documentclass[12pt,a4paper,english]{article}
\usepackage[ansinew]{inputenc}
\usepackage[T1]{fontenc}
\usepackage[english]{babel}
\usepackage{tabularx}
\usepackage{fancybox}
\usepackage{fancyhdr}
\usepackage{color}
\usepackage[numbers, sort]{natbib}

\bibpunct{[}{]}{,}{n}{,}{,}
\usepackage{amsmath}
\usepackage{amstext}
\usepackage{amssymb}
\usepackage{amsthm}
\usepackage{mathrsfs}
\newcommand{\be}{\begin{equation}}
\newcommand{\ee}{\end{equation}}
\newcommand{\ba}{\begin{aligned}}
\newcommand{\ea}{\end{aligned}}
\newcommand{\R}{\mathbb{R}}
\newcommand{\N}{\mathbb{N}}
\newcommand{\ind}{\mathbf{1}}
\newcommand{\lsi}{\left[\negthinspace\left[}
\newcommand{\rsi}{\right]\negthinspace\right]}
\newcommand{\lsir}{\left(\negthinspace\left(}
\newcommand{\rsir}{\right)\negthinspace\right)}
\newcommand{\M}{\mathcal{M}}
\newcommand{\A}{\mathcal{A}}
\newcommand{\F}{\mathcal{F}}
\newcommand{\FF}{\mathbb{F}}
\newcommand{\G}{\mathcal{G}}
\newcommand{\GG}{\mathbb{G}}
\newcommand{\HH}{\mathbb{H}}
\newcommand{\barS}{\bar{S}}
\newcommand{\tildeS}{\widetilde{S}}
\newcommand{\tildeX}{\widetilde{X}}
\newcommand{\tildeA}{\widetilde{A}}
\newcommand{\eps}{\varepsilon}
\newtheorem{Thm}{\bf Theorem}[section]
\newtheorem{Def}[Thm]{\bf Definition}
\newtheorem{Prop}[Thm]{\bf Proposition}
\newtheorem{Lem}[Thm]{\bf Lemma}
\newtheorem{Cor}[Thm]{\bf Corollary}
\theoremstyle{remark}
\newtheorem{Rem}[Thm]{\bf Remark}
\newtheorem{Rems}[Thm]{\bf Remarks}
\newtheorem{Ass}{\bf Assumption}
\newtheorem{Ex}[Thm]{\bf Example}

\numberwithin{equation}{section}
\usepackage{geometry}
\makeatletter
\renewcommand*\@fnsymbol[1]{\the#1}
\makeatother

\title{On arbitrages arising from honest times}

\author{Claudio Fontana\thanks{Corresponding author. INRIA Paris - Rocquencourt, Domaine de Voluceau, Rocquencourt, BP 105, Le Chesnay Cedex, 78153, France. E-mail: \texttt{claudio.fontana@inria.fr}.}
\and Monique Jeanblanc\thanks{Université d'Evry Val d'Essonne, Laboratoire Analyse et Probabilités, 23 Bd. de France, 91037, Evry (France). E-mail: \texttt{monique.jeanblanc@univ-evry.fr}.}
\and Shiqi Song\thanks{Université d'Evry Val d'Essonne, Laboratoire Analyse et Probabilités, 23 Bd. de France, 91037, Evry (France). E-mail: \texttt{shiqi.song@univ-evry.fr}.}}

\date{This version: July 18, 2013}

\begin{document}

\maketitle

\abstract{In the context of a general continuous financial market model, we study whether the additional information associated with an \emph{honest time} $\tau$ gives rise to arbitrage profits. By relying on the theory of progressive enlargement of filtrations, we explicitly show that no kind of arbitrage profit can ever be realised strictly before $\tau$, while classical arbitrage opportunities can be realised exactly at $\tau$ as well as after $\tau$. Moreover, stronger arbitrages of the first kind can only be obtained by trading as soon as $\tau$ occurs. We carefully study the behavior of local martingale deflators and consider no-arbitrage-type conditions weaker than NFLVR.}\\




\section{Introduction and motivation}	\label{S1}

The study of \emph{insider trading} behavior represents a classical issue in mathematical finance and financial economics. Loosely speaking, insider trading phenomena occur when agents having access to different information sets operate in the same financial market. In particular, the better informed agents may try to realise profits by relying on their deeper private knowledge and trading with the less informed agents. Typically, the presence of two distinct layers of information is mathematically represented by two filtrations $\FF=\left(\F_t\right)_{t\geq 0}$ and $\GG=\left(\G_t\right)_{t\geq 0}$ with $\F_t\subseteq\G_t$ for all $t\geq 0$. Intuitively, the filtration $\GG$ represents the information in possession of the insider trader. Assuming that the less informed agents cannot realise arbitrage profits by trading in the market, the fundamental question can be formulated in the following terms: is it possible for the insider trader to realise arbitrage profits by making use of the information contained in the larger filtration $\GG$? And, if yes, what is the appropriate notion of ``arbitrage profit'' and what is the trading strategy which yields that arbitrage profit?

The main goal of the present paper is to give complete and precise answers to the above questions in the context of a general continuous financial market model where the information of the insider is associated to an \emph{honest time} $\tau$. Referring to Section \ref{S2} for a precise definition of the notion of honest time (which seems to have been first introduced by \citet{MSW}), we would like to quote the following passage from \citet{DMM} (page 137) which intuitively explains the concept and seems particularly well suited to the present discussion:

\begin{quote}
\emph{Par exemple $S_t$ peut représenter le cours d'une certaine action à l'instant $t$, et $\tau$ est le moment idéal pour vendre son paquet d'actions. Tous les spéculateurs cherchent à conna\^itre $\tau$ sans jamais y parvenir, d'où son nom de variable al\'eatoire honn\^ete\footnote{We provide an English translation for the convenience of the reader: ``\emph{For instance, $S_t$ may represent the price of some stock at time $t$ and $\tau$ is the optimal time to liquidate a position in that stock. Every speculator strives to know when $\tau$ will occur, without ever achieving this goal. Hence, the name of honest random variable}''.}.}
\end{quote}

We consider a filtered probability space $\left(\Omega,\F,\FF,P\right)$ and let $S=\left(S_t\right)_{t\geq 0}$ represent the discounted price process of some risky assets. The filtration $\GG=\left(\G_t\right)_{t\geq 0}$ is constructed as the \emph{progressive enlargement} of $\FF$ with respect to an honest time $\tau$, which is assumed to avoid all $\FF$-stopping times. For a detailed account of the theory of progressive enlargement of filtrations, we refer the reader to Chapitres IV-V of \citet{Jeu} (see also Section 5.9.4 of \citet{JYC}, Section 9.2 of \citet{Nik} and Section VI.3 of \citet{Pr} for more rapid accounts and the book \citet{MY} for a presentation of the theory in the case where $\FF$ is a Brownian filtration). In this context, the investors who have  access only to the information contained in the filtration $\FF$ represent the ``spéculateurs'' referred to in the passage quoted above. In particular, they are not allowed to construct portfolio strategies based on $\tau$, simply because $\tau$ is not an $\FF$-stopping time. In contrast, an insider trader having access to the full information of the filtration $\GG$ can rely on his private knowledge about $\tau$ when trading in the market and, hence, may have the possibility of realising arbitrage profits.

In the present paper, we provide a complete analysis of the different kinds of arbitrage that can be realised by an insider trader having access to the additional information generated by an honest time. We do not confine ourselves to the classical no-arbitrage theory based on the notions of \emph{Arbitrage Opportunity} and \emph{Free Lunch with Vanishing Risk}, as developed by \citet{DS1}, but we also consider several stronger notions of arbitrage, namely \emph{Unbounded Increasing Profits}, \emph{Arbitrages of the First Kind} and \emph{Unbounded Profits with Bounded Risk}, which are of current interest in mathematical finance, as documented by the recent papers \citet{HS}, \citet{KK}, \citet{Kar,Kar3}, \citet{Shiqi2} and \citet{Tak}. In particular, this allows us to make precise the severity of the arbitrages induced by an honest time $\tau$. Furthermore, and this is a major aspect of the present paper, we carefully distinguish what kinds of arbitrage can be realised \emph{before}, \emph{at} and \emph{after} time $\tau$ and show that arbitrage profits incompatible with market viability can only be realised by trading as soon as $\tau$ occurs. In that sense, the present paper is to the best of our knowledge the first systematic study of the relations existing between progressive enlargements of filtrations with respect to honest times and no-arbitrage-type conditions.

It is already known that an honest time $\tau$ induces arbitrage opportunities in the progressively enlarged filtration $\GG$ immediately \emph{after} time $\tau$, see e.g. \citet{Imk} and \citet{Zw}. In comparison with these papers, our results provide two main innovations. On the one hand, we show that an insider trader can always realise an arbitrage opportunity not only \emph{after} $\tau$ but also exactly \emph{at} time $\tau$. On the other hand, we can explicitly exhibit in a simple way the trading strategies which yield the arbitrage profits. This contrasts with the approach adopted in \citet{Imk} and \citet{Zw}, where the existence of arbitrage opportunities is proved by relying on the abstract results of \citet{DS2}. Moreover, our approach permits to recover the results obtained in \citet{Imk} and \citet{Zw} in a very simple way. A key tool in our analysis is represented by the multiplicative decomposition of the Azéma supermartingale $Z=\left(Z_t\right)_{t\geq 0}$ associated to an honest time $\tau$, established in \citet{NY} (see Section \ref{S2} for precise definitions).

\begin{Ex}	\label{example}
We illustrate the framework of the present paper in the simplest possible setting (detailed proofs and sharper results in a general setting will be given in Sections \ref{S4}-\ref{S6}). Let $W=\left(W_t\right)_{t\geq 0}$ be a one-dimensional Brownian motion on the filtered probability space $(\Omega,\F,\FF^W\!,P)$, where $\FF^W=(\F^W_t)_{t\geq 0}$ denotes the natural filtration of $W$ (augmented by the $P$-nullsets of $\F^W_{\infty}$) and $\F:=\F^W_{\infty}$. Let the process $S=\left(S_t\right)_{t\geq 0}$ represent the discounted price of a risky asset and be given as the solution to the following SDE, for some $\bar{\sigma}\neq0$ and $s\in(0,\infty)$:
\be	\label{BS}	\left\{	\ba
dS_t &= S_t\,\bar{\sigma}\,dW_t,	\\
S_0 &= s.
\ea	\right. \ee
We define the finite random time $\tau:=\sup\left\{t\geq 0:S_t=\sup_{u\geq 0}S_u\right\}$ 
and the filtration $\GG=\left(\G_t\right)_{t\geq 0}$ as the progressive enlargement of $\FF$ with respect to $\tau$ (see Section \ref{S2} for precise definitions). 
\end{Ex} 

We call \emph{informed agent} an agent who can invest in the risky asset $S$ and has access to the information contained in the enlarged filtration $\GG$. The main results of the present paper can then be essentially summarised as follows:
\begin{itemize}
\item[\textbf{(a)}]
an informed agent can always realise arbitrage opportunities \emph{exactly at} time $\tau$, i.e., on the time interval $\left[0,\tau\right]$ (see Theorem \ref{attau}); 
\item[\textbf{(b)}]
it is never possible to realise arbitrage opportunities \emph{strictly before} time $\tau$, i.e., on the time interval $\left[0,\varrho\right]$, for any $\GG$-stopping time $\varrho$ with $\varrho<\tau$ $P$-a.s. (see Corollary \ref{cor-beforetau}). Furthermore, in the case of Example \ref{example}, an informed agent can never realise arbitrage opportunities on the time interval $\left[0,\tau\wedge T\right]$, for every $T\in\left(0,\infty\right)$ (see Theorem \ref{beforetau});
\item[\textbf{(c)}]
an informed agent can always realise arbitrage profits which are stronger than arbitrage opportunities (namely, arbitrages of the first kind, see Definition \ref{def-arb}) by taking a position \emph{as soon as} time $\tau$ has occurred (see Proposition \ref{arb2}) as well as arbitrage opportunities \emph{after} time $\tau$ (see Proposition \ref{arb-tau+eps}).
\end{itemize}
Furthermore, we can explicitly construct the trading strategies which realise the arbitrage profits for the informed agent in (a) and (c): it will be enough to hold appropriate long and short positions, respectively, in the portfolio which replicates the non-negative $\FF$-local martingale $N=\left(N_t\right)_{t\geq 0}$ appearing in the multiplicative decomposition of the Azéma supermartingale $Z=\left(Z_t\right)_{t\geq 0}$ of the random time $\tau$ (see Lemma \ref{NikYor}). In the specific case of Example \ref{example}, these positions will simply reduce to long and short buy-and-hold positions, respectively, in the risky asset $S$ itself (see Example \ref{rem-BS} and part 2 of Remarks \ref{rem-BS-3}).

The study of the impact of the additional information associated to a random time on the no-arbitrage-type properties of a financial market and on the behavior of market participants has already attracted attention in the mathematical finance literature. In particular, \citet{Imk} and \citet{Zw} are the closest precursors to our work (related results also appear in \citet{AI}). In the context of credit risk modelling, a study of the no-arbitrage-type properties of a market model with a filtration progressively enlarged with respect to a random time has also been recently undertaken in \citet{CJN}. We also want to mention that, in the case of \emph{initially} enlarged filtrations (see \citet{Jeu}, Chapitre III, or \citet{Pr}, Section VI.2), the possibility of realising arbitrage profits has been studied in \citet{GP} and \citet{IPW}. Finally, we refer the interested reader to \citet{NP1,NP2} for a detailed analysis of the role of honest times in financial modelling.

The paper is structured as follows. Section \ref{S2} describes the general setting and recalls several no-arbitrage-type conditions as well as some key technical results from the theory of progressive enlargement of filtrations. 
Sections \ref{S3} studies the existence and the properties of local martingale deflators in the progressively enlarged filtration $\GG$ up to different random time horizons, i.e., \emph{at}, \emph{before} and \emph{after} an honest time $\tau$. Sections \ref{S4}, \ref{S5} and \ref{S6} contain the main results on the existence of arbitrage profits with respect to the filtration $\GG$ \emph{at} $\tau$, \emph{before} $\tau$ and \emph{after} $\tau$, respectively. Section \ref{S7} concludes by discussing the role played by the standing Assumptions \ref{ass-NFLVR}-\ref{ass-PRP} introduced in Section \ref{S2} and by pointing out some possible generalisations of the results contained in the present paper.

\section{General setting and preliminary results}	\label{S2}

Let $(\Omega,\F,P)$ be a given probability space endowed with a filtration $\FF=(\F_t)_{t\geq 0}$ satisfying the usual conditions, where $P$ denotes the physical probability measure and $\F:=\F_{\infty}$. We consider a financial market comprising $d+1$ assets, with prices described by the $\R^{d+1}$-valued process $\barS=(\barS_t)_{t\geq 0}$. To allow for greater generality, we consider a financial market model on an infinite time horizon. Of course, financial markets on a finite time horizon $\left[0,T\right]$ can be embedded by simply considering the stopped process $(\barS)^T$. We assume that $\barS^{0}$ represents a numéraire or reference asset and is $P$-a.s. strictly positive. Without loss of generality, we express the prices of all $d+1$ assets in terms of $\barS^{0}$-discounted quantities, thus obtaining the $\R^d$-valued process $S=\left(S_t\right)_{t\geq 0}$, with $S^i:=\barS^i/\barS^{0}$ for each $i=1,\ldots,d$. We assume that the process $S$ is a \emph{continuous semimartingale} on $(\Omega,\F,\FF,P)$.

Let the random time $\tau:\Omega\rightarrow[0,\infty]$ be a $P$-a.s. finite \emph{honest time} on $(\Omega,\F,\FF,P)$. This means that $\tau$ is an $\F$-measurable random variable such that, for all $t>0$, there exists an $\F_t$-measurable random variable $\zeta_t$ with $\tau=\zeta_t$ on $\{\tau<t\}$ (see e.g. \citet{Jeu}, Chapitre V). We define the filtration $\GG=\left(\G_t\right)_{t\geq 0}$ as the \emph{progressive enlargement} of $\FF$ with respect to $\tau$, i.e., $\G_t:=\bigcap_{s>t}\bigl(\F_s\vee\sigma\left(\tau\wedge s\right)\bigr)$ for all $t\geq 0$, augmented by the $P$-nullsets of $\F=\F_{\infty}=\G_{\infty}$. It is well-known that $\GG$ is the smallest filtration satisfying the usual conditions which contains $\FF$ and makes $\tau$ a $\GG$-stopping time. Furthermore, the \emph{$\left(H'\right)$-hypothesis} holds between $\FF$ and $\GG$, meaning that any $\FF$-semimartingale is also a $\GG$-semimartingale (see \citet{Jeu}, Théorème 5.10). In particular, this implies that the discounted price process $S$ is also a $\GG$-semimartingale.

In order to model the activity of trading, we need to define the notion of admissible trading strategy, following \citet{DS1}. Let $\HH$ denote a generic filtration, i.e., in our setting $\HH\in\left\{\FF,\GG\right\}$. We denote by $L^{\HH}\left(S\right)$ the set of all $\R^d$-valued $\HH$-predictable processes $\theta=\left(\theta_t\right)_{t\geq 0}$ which are $S$-integrable in $\HH$ and we write $\theta\cdot S$ for the corresponding stochastic integral process.

\begin{Def}	\label{def-adm}
Let $\HH\in\left\{\FF,\GG\right\}$. For $a\in\R_+$, an element $\theta\in L^{\HH}\left(S\right)$ is said to be an \emph{$a$-admissible $\HH$-strategy} if $(\theta\cdot S)_{\infty}:=\lim_{t\rightarrow\infty}(\theta\cdot S)_t$ exists and $(\theta\cdot S)_t\geq-a$ $P$-a.s. for all $t\geq 0$. We denote by $\A^{\HH}_a$ the set of all $a$-admissible $\HH$-strategies. We say that an element $\theta\in L^{\HH}\left(S\right)$ is an \emph{admissible $\HH$-strategy} if $\theta\in\A^{\HH}:=\bigcup_{a\in\R_+}\!\A^{\HH}_a$.
\end{Def}

We assume that there are no frictions or trading constraints and that trading is done in a self-financing way. This implies that the wealth process generated by trading according to an admissible $\HH$-strategy $\theta$ starting from an initial endowment of $x\in\R$ is given by $V\left(x,\theta\right):=x+\theta\cdot S$, for $\HH\in\left\{\FF,\GG\right\}$. We call \emph{restricted} financial market the tuple $\M^{\FF}:=\left(\Omega,\F,\FF,P\,;S,\A^{\FF}\right)$, as opposed to the \emph{enlarged} financial market $\M^{\GG}:=\left(\Omega,\F,\GG,P\,;S,\A^{\GG}\right)$. Intuitively, agents operating in the enlarged financial market are better informed than agents operating in the restricted financial market, due to the additional information generated by the random time $\tau$.

\begin{Rem}	\label{rem-filtr}
Note that, since $\FF\subseteq\GG$ and all $\FF$-semimartingales are also $\GG$-semimartin\-gales, we have $L^{\FF}\left(S\right)\subseteq L^{\GG}\left(S\right)$, as can be deduced from Proposition 2.1 of \citet{Jeu}. In turn, this implies that $\A^{\FF}\subseteq\A^{\GG}$, thus reflecting the fact that agents in the enlarged financial market are allowed to use a richer information set to construct their portfolios.
\end{Rem}

As mentioned in the introduction, the present paper aims at answering the following question: how does the additional information associated to the honest time $\tau$ give rise to arbitrage profits? To this end, let us first recall three important notions of arbitrage which have appeared in the literature.

\begin{Def}	\label{def-arb}
Let $\HH\in\{\FF,\GG\}$.
\begin{itemize}
\item[(i)]
A non-negative $\F$-measurable random variable $\xi$ with $P\left(\xi>0\right)>0$ yields an \emph{Arbitrage of the First Kind} if for all $x>0$ there exists an element $\theta^x\in\A^{\HH}_x$ such that $V\left(x,\theta^x\right)_{\infty}\geq\xi$ $P$-a.s.
If there exists no such random variable we say that  the financial market $\M^{\HH}$ satisfies the \emph{No Arbitrage of the First Kind (NA1)} condition.
\item[(ii)]
An element $\theta\in\A^{\HH}$ yields an \emph{Arbitrage Opportunity} if $V\left(0,\theta\right)_{\infty}\geq 0$ $P$-a.s. and $P\bigl(V\left(0,\theta\right)_{\infty}>0\bigr)>0$. If there exists no such $\theta\in\A^{\HH}$ we say that the financial market $\M^{\HH}$ satisfies the \emph{No Arbitrage (NA)} condition.
\item[(iii)]
A sequence $\left\{\theta^n\right\}_{n\in\N}\subset\A^{\HH}$ yields a \emph{Free Lunch with Vanishing Risk} if there exist an $\varepsilon>0$ and an increasing sequence $\left\{\delta_{n}\right\}_{n\in\N}$ with $0\leq\delta_{n}\nearrow 1$ such that $P\bigl(V\left(0,\theta^n\right)_{\infty}>-1+\delta_{n}\bigr)=1$ and $P\bigl(V\left(0,\theta^n\right)_{\infty}>\varepsilon\bigr)\geq\varepsilon$. If there exists no such sequence we say that the financial market $\M^{\HH}$ satisfies the \emph{No Free Lunch with Vanishing Risk (NFLVR)} condition.
\end{itemize}
For a (possibly infinite-valued) $\HH$-stopping time $\varrho$, we say that NA1/NA/NFLVR holds in the financial market $\M^{\HH}$ on the time horizon $\left[0,\varrho\right]$ if the financial market $\left(\Omega,\F,\HH,P\,;S^{\varrho},\A^{\HH}\right)$ satisfies NA1/NA/NFLVR, where $S^{\varrho}$ denotes the stopped process $(S_{t\wedge\varrho})_{t\geq 0}$.
\end{Def}

The notion of \emph{No Arbitrage of the First Kind} has been recently introduced by \citet{Kar} and can be shown to be equivalent to the boundedness in probability of the set $\left\{V(1,\theta)_{\infty}:\theta\in\A^{\HH}_1\right\}$, see Proposition 1 of \citet{Kar}. The latter condition has appeared under the name of \emph{No Unbounded Profit with Bounded Risk (NUPBR)} in \citet{KK} but its importance was first recognised by \citet{DS1} and \citet{Kab} (see also \citet{Fon2} for a discussion of the relations between the different conditions introduced in Definition \ref{def-arb}).
The NA1 and NFLVR conditions can both be characterised in purely probabilistic terms. As a preliminary, let us recall the following definition.

\begin{Def}	\label{defl}
Let $\HH\in\{\FF,\GG\}$ and $\varrho$ a (possibly infinite-valued) $\HH$-stopping time.
\begin{itemize}
\item[(i)]
A strictly positive $\HH$-local martingale $L=(L_t)_{t\geq 0}$ with $L_0=1$ and $L_{\infty}>0$ $P$-a.s. is said to be a \emph{local martingale deflator in $\HH$ on the time horizon $\left[0,\varrho\right]$} if the process $L\,S^{\varrho}$ is an $\HH$-local martingale;
\item[(ii)]
a probability measure $Q\sim P$ on $(\Omega,\F)$ is said to be an \emph{Equivalent Local Martingale Measure in $\HH$ (ELMM$_{\HH}$) on the time horizon $\left[0,\varrho\right]$} if the process $S^{\varrho}$ is an $\HH$-local martingale under $Q$.
\end{itemize}
\end{Def}

Note that the concept of local martingale deflator corresponds to the notion of \emph{strict martingale density} first introduced by \citet{Sch1}. We then have the following fundamental theorem. The first assertion is a partial statement of Theorem 4 of \citet{Kar} (noting that the proof carries over to the infinite time horizon case), while the last two assertions follow from Corollary 1.2 and Corollary 3.4 together with Proposition 3.6, respectively, of \citet{DS1}.

\begin{Thm}	\label{char-arb}
Let $\HH\in\{\FF,\GG\bigr\}$ and $\varrho$ a (possibly infinite-valued) $\HH$-stopping time. Then, on the time horizon $\left[0,\varrho\right]$, the following hold:
\begin{itemize}
\item[(i)]
NA1 (or, equivalently, NUPBR) holds in the financial market $\M^{\HH}$ if and only if there exists a local martingale deflator in $\HH$;
\item[(ii)]
NFLVR holds in the financial market $\M^{\HH}$ if and only if there exists an Equivalent Local Martingale Measure in $\HH$.
\item[(iii)]
NFLVR holds in the financial market $\M^{\HH}$ if and only if both NA1 (or, equivalently, NUPBR) and NA hold in the financial market $\M^{\HH}$.
\end{itemize}
\end{Thm}

In particular, the above theorem implies that NFLVR holds in the financial market $\M^{\HH}$ on the time horizon $[0,\varrho]$ if and only if there exists a local martingale deflator $L$ in $\HH$ such that $L^{\varrho}$ is a uniformly integrable $\HH$-martingale. We shall always work under the following standing assumption, which ensures that the restricted financial market $\M^{\FF}$ does not allow for any kind of arbitrage\footnote{We want to point out that analogous results can be obtained if in Assumption \ref{ass-NFLVR} the NFLVR condition is replaced with the weaker NA1 (or, equivalently, NUPBR) condition (on this regard, see also the discussion in Section \ref{S7}). We also want to make the reader aware of the fact that Assumption \ref{ass-NFLVR} excludes the classical \citet{BS} model with a non-zero drift coefficient on an infinite time horizon, see e.g. Example 1.7.6 in \citet{KS}.}.

\begin{Ass}	\label{ass-NFLVR}
The restricted financial market $\M^{\FF}$ satisfies NFLVR.
\end{Ass}

We aim at studying the no-arbitrage-type properties (or the lack thereof) of the enlarged financial market $\M^{\GG}$. In the remaining part of the paper, we shall give a clear answer to this issue under the two following standing assumptions, where we denote by $M=\left(M_t\right)_{t\geq 0}$ the $\FF$-local martingale part in the canonical decomposition of the semimartingale $S$ in the filtration $\FF$.

\begin{Ass}	\label{ass-avoid}
The random time $\tau$ avoids all $\FF$-stopping times: for every $\FF$-stopping time $T$ we have $P\left(\tau=T\right)=0$.
\end{Ass}

\begin{Ass}	\label{ass-PRP}
The continuous $\FF$-local martingale $M=\left(M_t\right)_{t\geq 0}$ has the $\FF$-predicta\-ble representation property in the filtration $\FF$. 
\end{Ass}

Assumption \ref{ass-avoid} is classical when dealing with progressive enlargements of filtrations (see e.g. \citet{JYC}, Section 5.9.4). Assumption \ref{ass-PRP} means that any $\FF$-local martingale $U=(U_t)_{t\geq 0}$ with $U_0=0$ can be represented as $U=\varphi\cdot M$, where $\varphi=(\varphi_t)_{t\geq 0}$ is an $\R^d$-valued $\FF$-predictable process such that $\int_0^t\varphi_s\,\!'d\langle M,M\rangle_s\varphi_s<\infty$ $P$-a.s. for all $t\geq 0$, see e.g. Chapter III of \citet{JacShi}. In particular, Assumption \ref{ass-PRP} implies that all $\FF$-martingales are continuous. We postpone to Section \ref{S5} a discussion of the importance of Assumptions \ref{ass-NFLVR}-\ref{ass-PRP} and of possible extensions thereof (compare also with Example \ref{rem-counterexample}).

\begin{Rem}[On the completeness of the restricted financial market]	\label{complete}
Under the additional assumption that the initial $\sigma$-field $\F_0$ is trivial, Assumptions \ref{ass-NFLVR} and \ref{ass-PRP} together imply that there exists a unique ELMM$_{\FF}$ $Q$ for the restricted financial market $\M^{\FF}$, see e.g. Theorem 9.5.3.1 of \citet{JYC}. In turn, this implies that for any non-negative $\F$-measurable random variable $\xi\in L^1\left(Q\right)$ there exists a strategy $\theta^{\xi}\in\A^{\FF}$ such that $\xi=x+(\theta^{\xi}\!\cdot\!S)_{\infty}$ $P$-a.s., for some $x\in\R$ and where $\theta^{\xi}\cdot S$ is a uniformly integrable $\left(Q,\FF\right)$-martingale (see e.g. \citet{DS1}, Theorem 5.2, or \citet{AS}, Théorème 3.2).
\end{Rem}

We close this section by recalling two technical results obtained by \citet{NY} under the hypothesis that all $\FF$-local martingales are continuous and Assumption \ref{ass-avoid} holds. Recall also that a $P$-a.s. finite random time $\tau$ is an honest time if and only if it is the end of an $\FF$-optional set (see \citet{Jeu}, Proposition 5.1) and note that, due to Assumption \ref{ass-PRP} together with the continuity of $S$, the $\FF$-optional sigma field coincides with the $\FF$-predictable sigma field. In the following, we denote by $Z=\left(Z_t\right)_{t\geq 0}$ the \emph{Azéma supermartingale} of the random time $\tau$, i.e., $Z_t=P\left(\tau>t|\F_t\right)$ for all $t\geq 0$.

\begin{Lem}[\citet{NY}, Theorem 4.1]	\label{NikYor}
There exists a continuous non-negative $\FF$-local martingale $N=(N_t)_{t\geq 0}$ with $N_0=1$ and $\lim_{t\rightarrow\infty}N_t=0$ $P$-a.s. such that $Z$ admits the following multiplicative decomposition, for all $t\geq 0$:
$$
Z_t = P\left(\tau>t|\F_t\right) = \frac{N_t}{N^*_t}
$$
where $N^*_t:=\sup_{s\leq t}N_s$. Furthermore, we have that:
$$
\tau = \sup\left\{t\geq 0:N_t=N^*_t\right\} = \sup\left\{t\geq 0:N_t=N^*_{\infty}\right\}.
$$
\end{Lem}

\begin{Lem}[\citet{NY}, Proposition 2.5]	\label{dec}
Every $\FF$-local martingale $X=\left(X_t\right)_{t\geq 0}$ has the following canonical decomposition as a semimartingale in $\GG$:
$$
X_t = \tildeX_t+\int_0^{t\wedge\tau}\frac{d\langle X,N\rangle_{\!s}}{N_s}
-\int_{\tau}^{t\vee\tau}\frac{d\langle X,N\rangle_{\!s}}{N^*_{\infty}-N_s}
$$
where $\tildeX=(\tildeX_t)_{t\geq 0}$ is a $\GG$-local martingale and $N=(N_t)_{t\geq 0}$ is as in Lemma \ref{NikYor}.
\end{Lem}

In the specific case of Example \ref{example}, the local martingale $N$ appearing in Lemma \ref{NikYor} is equal to $S/S_0$ itself (see Example \ref{rem-BS}). In our context, the importance of Lemma \ref{NikYor} consists in the possibility of reducing the study of the existence of arbitrage profits in a general enlarged financial market $\M^{\GG}$ to the simple situation considered in Example \ref{example}, as will be shown in the remaining part of the paper (in particular, see Examples \ref{last passage} and \ref{drawdown}).

\section{Existence and properties of local martingale deflators in $\GG$}	\label{S3}

In view of Theorem \ref{char-arb}, local martingale deflators play a fundamental role in characterising the validity of the NA1 and NFLVR conditions. The goal of the present section consists in studying the existence and the properties of local martingale deflators in the progressively enlarged filtration $\GG$. More specifically, in Section \ref{S3.1} we shall study local martingale deflators in $\GG$ on the random time horizon $[0,\tau]$, while in Section \ref{S3.2} we shall restrict our attention to the random time horizon $[0,\sigma\wedge\tau]$, for an arbitrary $\FF$-stopping time $\sigma$. Finally, Section \ref{S3.3} deals with the existence of local martingale deflators in $\GG$ on the global time horizon $[0,\infty]$.

Without any loss of generality, we may and do assume that $P$ is already an ELMM$_{\FF}$ for the restricted financial market $\M^{\FF}$. Indeed, Assumption \ref{ass-NFLVR} together with part \emph{(ii)} of Theorem \ref{char-arb} ensures the existence of an ELMM$_{\FF}$ $Q$. Since $Q\sim P$, it is easy to verify that all the properties of the general setting described in Section \ref{S2} still hold under $Q$. More precisely, the random time $\tau$ is still an honest time which avoids $\FF$-stopping times under any ELMM$_{\FF}$ $Q$ and the $\left(Q,\FF\right)$-local martingale $S=\left(S_t\right)_{t\geq 0}$ has the $\FF$-predictable representation property under the measure $Q$ (see \citet{JacShi}, part \emph{a)} of Theorem III.5.24). Hence, Assumptions \ref{ass-NFLVR}, \ref{ass-avoid} and \ref{ass-PRP} hold under every ELMM$_{\FF}$ $Q$. Finally, observe that the notion of admissible strategy given in Definition \ref{def-adm} is stable under an equivalent change of measure (see e.g. \citet{JacShi}, Proposition III.6.24). As a consequence, all the NA1, NA and NFLVR no-arbitrage-type conditions introduced in Definition \ref{def-arb} hold for the enlarged financial market $\M^{\GG}$ under the measure $P$ if and only if they hold under the measure $Q$. 

\subsection{Local martingale deflators in $\GG$ on the time horizon $[0,\tau]$}	\label{S3.1}

Note first that, due to Assumption \ref{ass-PRP}, the $\FF$-local martingale $N$ appearing in Lemma \ref{NikYor} admits the stochastic integral representation $N=1+\varphi\cdot S$, where $\varphi=(\varphi_t)_{t\geq 0}$ is an $\R^d$-valued $\FF$-predictable process in $L^{\FF}\left(S\right)$. Furthermore, by Lemma \ref{dec}, the stopped process $S^{\tau}$ admits the following canonical decomposition in the filtration $\GG$:
\be	\label{candec}
S^{\tau}_t = \tildeS_t^{\tau}+\int_0^{t\wedge\tau}\frac{d\langle S,N\rangle_{\!s}}{N_s}
= \tildeS^{\tau}_t+\int_0^{t\wedge\tau}\!d\langle S,S\rangle_{\!s}\frac{\varphi_{s}}{N_s}
= \tildeS^{\tau}_t+\int_0^td\langle \tildeS^{\tau},\tildeS^{\tau}\rangle_{\!s}\frac{\varphi_{s}}{N_s}
\ee
where $\tildeS=(\tildeS_t)_{t\geq 0}$ is a continuous $\GG$-local martingale. 

\begin{Prop}	\label{exdefl}
The process $1/N^{\tau}=\left(1/N_{t\wedge\tau}\right)_{t\geq 0}$ is a local martingale deflator in $\GG$ on the time horizon $\left[0,\tau\right]$. Furthermore, the process $1/N^{\tau}$ fails to be a uniformly integrable $\GG$-martingale.
\end{Prop}
\begin{proof}
We define the $\FF$-stopping time $\nu:=\inf\left\{t\geq 0:N_t=0\right\}=\inf\left\{t\geq 0:Z_t=0\right\}$, where the process $Z=(Z_t)_{t\geq 0}$ is the Azéma supermartingale associated to $\tau$. Noting that $P\left(\tau>\nu|\F_{\nu}\right)=Z_{\nu}=0$, we get $\tau\leq\nu$ $P$-a.s. Since $\tau$ avoids all $\FF$-stopping times (Assumption \ref{ass-avoid}), we furthermore have $\tau<\nu$ $P$-a.s. Thus, the process $1/N^{\tau}$ is well-defined. By It\^o's formula together with equation \eqref{candec}:
\be	\label{1/N}	\ba
\frac{1}{N^{\tau}} &= 1-\frac{1}{\left(N^{\tau}\right)^2}\cdot N^{\tau}+\frac{1}{\left(N^{\tau}\right)^3}\cdot\langle N\rangle^{\tau}
= 1-\frac{\varphi}{\left(N^{\tau}\right)^2}\cdot S^{\tau}
+\frac{\varphi}{\left(N^{\tau}\right)^3}\cdot\langle S^{\tau},N\rangle	\\
&= 1-\frac{\varphi}{\left(N^{\tau}\right)^2}\cdot\tildeS^{\tau}.
\ea	\ee
This shows that $1/N^{\tau}$ is a strictly positive continuous $\GG$-local martingale satisfying $1/N^{\tau}_0=1$ and $1/N^{\tau}_{\infty}=1/N_{\tau}>0$ $P$-a.s. Furthermore, using the product rule together with equations \eqref{candec}-\eqref{1/N}:
$$
\frac{S^{\tau}}{N^{\tau}}
= S_0+\frac{1}{N^{\tau}}\cdot S^{\tau}
+S^{\tau}\cdot\frac{1}{N^{\tau}}
+\Bigl\langle S^{\tau},\frac{1}{N^{\tau}}\Bigr\rangle
= S_0+\frac{1}{N^{\tau}}\cdot\tildeS^{\tau}
+S^{\tau}\cdot\frac{1}{N^{\tau}}\,.
$$
This shows that $1/N^{\tau}$ is a local martingale deflator in $\GG$ on the time horizon $\left[0,\tau\right]$. Being a positive $\GG$-local martingale, the process $1/N^{\tau}$ is also a $\GG$-supermartingale by Fatou's lemma. It is a uniformly integrable $\GG$-martingale if and only if $E\left[1/N^{\tau}_{\infty}\right]=E\left[1/N_{\tau}\right]=E\left[1/N_0\right]=1$. However, $N_{\tau}\geq 1$ $P$-a.s. and $P\left(N_{\tau}>1\right)>0$ imply that $E\left[1/N_{\tau}\right]<1$.
\end{proof}

Proposition \ref{exdefl} shows that there always exists at least one local martingale deflator in $\GG$ on the time horizon $\left[0,\tau\right]$, given by the reciprocal of the $\FF$-local martingale $N$ appearing in the multiplicative decomposition of the Azéma supermartingale $Z$ associated to the random time $\tau$ (see Lemma \ref{NikYor}).

\begin{Rems}
1) 
According to the terminology adopted in \citet{HS}, the process $1/N^{\tau}=\mathcal{E}\bigl(-\frac{\varphi}{N}\cdot\tildeS^{\tau}\bigr)$ represents the \emph{minimal martingale density} for the stopped process $S^{\tau}$ in the progressively enlarged filtration $\GG$, i.e., the density process of the \emph{minimal martingale measure} (when the latter exists).

2)
As long as all $\FF$-local martingales are continuous, the result of Proposition \ref{exdefl} can be readily extended to the case where $\tau$ is an arbitrary random time (i.e., not necessarily honest) satisfying Assumption \ref{ass-avoid}. Indeed, as shown in the proof of Proposition \ref{exdefl}, we have $Z>0$ on $\lsi0,\tau\rsi$ and, hence, the supermartingale $Z$ admits the multiplicative decomposition $Z=N/D$, where $N=(N_t)_{t\geq0}$ is an $\FF$-local martingale with $N^{\tau}>0$ $P$-a.s. and $D=(D_t)_{t\geq 0}$ is an $\FF$-predictable increasing process (see \citet{JacShi}, Theorem II.8.21). The same arguments as in the proof of Proposition \ref{exdefl} allow then to show that $1/N^{\tau}$ is a local martingale deflator in $\GG$ on $[0,\tau]$.
\end{Rems}

The next lemma describes the general structure of all local martingale deflators in $\GG$ on the time horizon $\left[0,\tau\right]$. Furthermore, it shows that all local martingale deflators in $\GG$ are \emph{strict} $\GG$-local martingales in the sense of \citet{DS3}, being $\GG$-local martingales which fail to be uniformly integrable $\GG$-martingales.

\begin{Lem}	\label{str-defl}
Let $L=(L_t)_{t\geq 0}$ be a local martingale deflator in $\GG$ on the time horizon $\left[0,\tau\right]$. Then $L$ admits the following representation:
$$
L = \frac{\mathcal{E}(R)}{N^{\tau}}
$$
where $R=\left(R_t\right)_{t\geq 0}$ is a $\GG$-local martingale with $R_0=0$, purely discontinuous on $\lsi0,\tau\rsi$ and with $\left\{\Delta R\neq 0\right\}\subseteq\lsi\tau\rsi$ and $\Delta R_{\tau}>-1$ $P$-a.s. Furthermore, all local martingale deflators in $\GG$ on the time horizon $\left[0,\tau\right]$ fail to be uniformly integrable $\GG$-martingales.
\end{Lem}
\begin{proof}
We already know from Proposition \ref{exdefl} that the set of all local martingale deflators in $\GG$ on the time horizon $\left[0,\tau\right]$ is non-empty. Let $L=\left(L_t\right)_{t\geq 0}$ be an element of that set. Theorem 1 of \citet{Sch2} (or also \citet{CS}, Théorème 2.2) together with equation \eqref{candec} shows that $L$ admits the representation
$$
L = \mathcal{E}\!\left(-\frac{\varphi}{N}\cdot\tildeS^{\tau}\right)\mathcal{E}(R)
$$
where $R=\left(R_t\right)_{t\geq 0}$ is a $\GG$-local martingale with $R_0=0$ and $\Delta R>-1$ $P$-a.s. such that $\langle R,\tildeS^i\rangle^{\tau}=0$ for all $i=1,\ldots,d$. The uniqueness of the Doléans-Dade exponential together with equation \eqref{1/N} implies that $1/N^{\tau}=\mathcal{E}\!\bigl(-\frac{\varphi}{N}\cdot\tildeS^{\tau}\bigr)$. Let $R=R^{c}+R^{d}$ be the decomposition of the $\GG$-local martingale $R$ into its continuous and purely discontinuous $\GG$-local martingale parts. Since $\langle R^{c}\!,\tildeS^i\rangle^{\tau}=\bigl\langle\left(R^{c}\right)^{\tau}\!\!,\tildeS^i\bigr\rangle=0$ for all $i=1,\ldots,d$ by orthogonality, Proposition 5.4 of \citet{Bar} implies that $\left(R^{c}\right)^{\tau}=0$, thus showing that $R$ is purely discontinuous on $\lsi0,\tau\rsi$. Furthermore, Théorème 5.12 of \citet{Jeu} implies that $\left\{\Delta R\neq 0\right\}\subseteq\lsi\tau\rsi$, since all $\FF$-local martingales are continuous. It remains to show that $L$ fails to be a uniformly integrable $\GG$-martingale. For that, it suffices to observe that:
$$
E\left[L_{\infty}\right]
= E\left[\frac{\mathcal{E}(R)_{\infty}}{N_{\tau}}\right]
< E\left[\mathcal{E}(R)_{\infty}\right] \leq 1
$$
where the first inequality follows since $N_{\tau}\geq 1$ $P$-a.s. and $P\left(N_{\tau}>1\right)>0$ and the last inequality is due to the supermartingale property of the positive $\GG$-local martingale $\mathcal{E}(R)$.
\end{proof}

\begin{Rem}	\label{rem-rep}
The structure of the $\GG$-local martingale $R$ appearing in Lemma \ref{str-defl} can be described a bit more explicitly by relying on the general martingale representation results recently established in \citet{JS} for progressively enlarged filtrations. Indeed, noting that the dual $\FF$-predictable projection of the process $(\ind_{\left\{\tau\leq t\right\}})_{t\geq 0}$ is given by $\bigl(\log\left(N^*_t\right)\bigr)_{t\geq 0}=\bigl(\int_0^t\!\!\frac{1}{N^*_s}\,dN^*_s\bigr)_{t\geq 0}$ (see \citet{NY}, Corollary 2.4, or also \citet{MY}, Exercise 1.8), Theorem 6.2 of \citet{JS} implies that the following representation holds true\footnote{Note that, as in Corollary III.4.27 of \citet{JacShi}, the martingale representation result obtained in Theorem 6.2 of \citet{JS} for bounded $\GG$-martingales extends naturally to all $\GG$-local martingales. The representation \eqref{rep-JS} then follows by Lemma \ref{str-defl} together with Theorem 6.2 of \citet{JS} and Theorem I.4.61 of \citet{JacShi}.}:
\be	\label{rep-JS}
L^{\tau} = \frac{1}{N^{\tau}}
\exp\left(-\frac{k}{N^*}\cdot N^*\right)
\left(1+k_{\tau}\ind_{\lsi\tau,\infty\rsir}+\eta\ind_{\lsi\tau,\infty\rsir}\right)
\ee
where $k=(k_t)_{t\geq 0}$ is an $\FF$-predictable process such that $\int_0^{\tau}\!\frac{|k_s|}{N^*_s}\,dN^*_s<\infty$ and $1+k_{\tau}>0$ $P$-a.s. and $\eta$ is a non-negative $\G_{\tau}$-measurable random variable with $E\left[\eta|\G_{\tau-}\right]=0$.
\end{Rem}

\subsection{Local martingale deflators in $\GG$ on the time horizon $[0,\sigma\wedge\tau]$}	\label{S3.2}

Let $\sigma$ be an arbitrary (possibly infinite-valued) $\FF$-stopping time. In this section, aiming at characterising the validity of NFLVR on $[0,\sigma\wedge\tau]$ (see Section \ref{S5}), we study the martingale property of local martingale deflators in $\GG$ on the time horizon $[0,\sigma\wedge\tau]$. 

For any $\FF$-stopping time $\sigma$, due to the proof of Lemma \ref{str-defl} together with Remark \ref{rem-rep}, every local martingale deflator $L=(L_t)_{t\geq 0}$ in $\GG$ on the time horizon $[0,\sigma\wedge\tau]$ admits the following representation when stopped at $\sigma\wedge\tau$:
\be	\label{str-defl-sigma}
L^{\sigma\wedge\tau} = \frac{1}{N^{\sigma\wedge\tau}}
\exp\left(-\frac{k\,\ind_{\lsi0,\sigma\rsi}}{N^*}\cdot N^*\right)
\Bigl(1+\ind_{\{\tau\leq\sigma\}}\!\left(k_{\tau}\ind_{\lsi\tau,\infty\rsir}+\eta\ind_{\lsi\tau,\infty\rsir}\right)\Bigr)
\ee
where $k=(k_t)_{t\geq 0}$ is an $\FF$-predictable process such that $\int_0^{\tau}\!\frac{|k_s|}{N^*_s}\,dN^*_s<\infty$ and $1+k_{\tau}>0$ $P$-a.s. and $\eta$ is a non-negative $\G_{\tau}$-measurable random variable with $E\left[\eta|\G_{\tau-}\right]=0$.

The proof of the following lemma (postponed to the Appendix) is technical and can be omitted on a first reading. Recall that $\nu=\inf\left\{t\geq 0:Z_t=0\right\}$, as in the proof of Proposition \ref{exdefl}. 

\begin{Lem}	\label{comp}
Let $\sigma$ be an $\FF$-stopping time and $L=\left(L_t\right)_{t\geq 0}$ a local martingale deflator in $\GG$ on the time horizon $\left[0,\sigma\wedge\tau\right]$. Then the following holds:
\be	\label{comp-a}
E\bigl[L_{\sigma\wedge\tau}\bigr] = 
E\left[1-\exp\left(-\int_0^{\tau}\frac{1+k_s}{N^*_s}\,dN^*_s\right)
\ind_{\left\{\nu\leq\sigma\right\}}\right]
\ee
where $k=\left(k_t\right)_{t\geq 0}$ is the $\FF$-predictable process appearing in the representation \eqref{str-defl-sigma} and $\int_0^{\tau}\!\frac{1+k_s}{N^*_s}\,dN^*_s>0$ $P$-a.s.
As a consequence, the stopped process $L^{\sigma\wedge\tau}$ is a uniformly integrable $\GG$-martingale if and only if $P(\nu\leq\sigma)=0$.
\end{Lem}

\begin{Rem}[On martingales and strict local martingales]
Let $L=(L_t)_{t\geq 0}$ be a local martingale deflator in $\GG$ on the time horizon $\left[0,\tau\right]$. By Fatou's lemma, the strictly positive $\GG$-local martingale $L^{\tau}$ is also a $\GG$-supermartingale. As a consequence, $L^{\tau}$ is a (true) $\GG$-martingale if and only if it has constant expectation, i.e., $E\left[L_{t\wedge\tau}\right]=1$ for all $t\geq 0$. Due to Lemma \ref{comp} (with $\sigma=t$), the latter condition holds if and only if $P\left(\nu=\infty\right)=1$. This means that, as soon as $P\left(\nu<\infty\right)>0$, any local martingale deflator $L$ in $\GG$ on $[0,\tau]$ is a \emph{strict} $\GG$-local martingale in the sense of \citet{ELY}, being a $\GG$-local martingale which fails to be a (true) $\GG$-martingale. It is interesting to note that Lemma \ref{str-defl} gives then a recipe for constructing a whole class of possibly discontinuous \emph{strict} $\GG$-local martingales (in the sense of \citet{ELY}). To the best of our knowledge, apart from the particular cases considered in \citet{Chy} and \citet{KKN}, there exist very few non-trivial examples of strict local martingales that are not necessarily continuous.
\end{Rem}

\subsection{Local martingale deflators in $\GG$ on the time horizon $[0,\infty]$}	\label{S3.3}

Let us now consider the question of whether there exists a local martingale deflator in $\GG$ on the global time horizon $[0,\infty]$. As a preliminary, recall that Lemma \ref{dec} gives the following canonical decomposition of $S=\left(S_t\right)_{t\geq 0}$ in the enlarged filtration $\GG$:
\be	\label{candec-2}
S_t = \tildeS_t + \int_0^{t\wedge\tau}\frac{d\langle S,N\,\rangle_{\!s}}{N_s}
-\int_{\tau}^{t\vee\tau}\frac{d\langle S,N\,\rangle_{\!s}}{N^*_{\infty}-N_s}
= \tildeS_t + \int_0^td\!\langle\tildeS,\tildeS\rangle_{\!s}\,\tilde{\alpha}_{s}
=: \tildeS_t+\tildeA_t
\ee
where $\tildeS=(\tildeS_t)_{t\geq 0}$ is a $\GG$-local martingale, $\tilde{\alpha}_t:=\ind_{\left\{t\leq\tau\right\}}\varphi_t/N_t-\ind_{\left\{t>\tau\right\}}\varphi_t/(N^*_{\infty}-N_t)$ and where the process $\varphi=(\varphi_t)_{t\geq 0}\in L^{\FF}\left(S\right)\subset L^{\GG}(S)$ is the integrand in the stochastic integral representation $N=1+\varphi\cdot S$, with $N$ as in Lemma \ref{NikYor}.

\begin{Prop}	\label{defl-2}
The enlarged financial market $\M^{\GG}$ does not admit any local martingale deflator in $\GG$ on the global time horizon $\left[0,\infty\right]$.
\end{Prop}
\begin{proof}
Suppose that $L=\left(L_t\right)_{t\geq 0}$ is a local martingale deflator in $\GG$ on the global time horizon $\left[0,\infty\right]$. Similarly as in the proof of Lemma \ref{str-defl}, Theorem 1 of \citet{Sch2} together with equation \eqref{candec-2} implies that $L$ admits the following representation:
$$
L=\mathcal{E}\bigl(-\tilde{\alpha}\cdot\tildeS\,\bigr)\,\mathcal{E}\left(R\right)
$$
where $R=\left(R_t\right)_{t\geq 0}$ is a purely discontinuous $\GG$-local martingale satisfying $R_0=0$, $\left\{\Delta R\neq 0\right\}\subseteq\lsi\tau\rsi$ and $\Delta R_{\tau}>-1$ $P$-a.s. By using the definition of $\tilde{\alpha}$ together with Lemma \ref{str-defl}, we can write as follows:
$$
L=\mathcal{E}\!\left(-\ind_{\lsi0,\tau\rsi}\frac{\varphi}{N}\cdot\tildeS\,\right)
\mathcal{E}\!\left(\ind_{\lsir\tau,\infty\rsir}\frac{\varphi}{N^*_{\infty}-N}\cdot\tildeS\right)
\mathcal{E}\!\left(R\right)
= \frac{\hat{L}}{N^{\tau}}\,\mathcal{E}\!\left(R\right)
$$
with $\hat{L}:=\mathcal{E}\bigl(\ind_{\lsir\tau,\infty\rsir}\frac{\varphi}{N^*_{\infty}-N}\cdot\tildeS\,\bigr)$. 
Lemma \ref{NikYor} implies that $Z_t<1$ $P$-a.s. for all $t>\tau$ (see also \citet{Bar}, Lemma 2.4). Hence, using It\^o's formula together with Lemma \ref{NikYor} and Lemma \ref{dec}, we can write as follows, for all for $\tau<s\leq t$:
$$	\ba
\frac{1}{1-Z_t}-\frac{1}{1-Z_s}
&= \frac{N^*_{\tau}}{N^*_{\tau}-N_t}-\frac{N^*_{\tau}}{N^*_{\tau}-N_s}
= \int_s^t\!\!\!\frac{N^*_{\tau}}{\left(N^*_{\tau}-N_u\right)^2}\,dN_u
+\int_s^t\!\!\!\frac{N^*_{\tau}}{\left(N^*_{\tau}-N_u\right)^3}\,d\langle N\rangle_u	\\
&= \int_s^t\!\frac{N^*_{\tau}}{\left(N^*_{\tau}-N_u\right)^2}\,\varphi_u\,dS_u
+\int_s^t\!\frac{N^*_{\tau}}{\left(N^*_{\tau}-N_u\right)^3}\,\varphi_u'\,d\langle S,N\rangle_u	\\
&= \int_s^t\!\frac{N^*_{\tau}}{\left(N^*_{\tau}-N_u\right)^2}\,\varphi_u\,d\tildeS_u
= \int_s^t\!\frac{1}{1-Z_{u}}\frac{\varphi_u}{N^*_{\tau}-N_u}\,d\tildeS_u\,.
\ea	$$
Recalling that $N^*_{\infty}=N^*_{\tau}$, the uniqueness of the Doléans-Dade exponential implies that $\hat{L}_t-\hat{L}_{s}=\frac{1}{1-Z_t}-\frac{1}{1-Z_s}$ for all $\tau<s\leq t$. So, we can write:
$$
\underset{s\searrow\,\tau}{\lim}\frac{1}{1-Z_s}
= \frac{1}{1-Z_t}-\hat{L}_t+\hat{L}_{\tau}<\infty
\qquad \text{$P$-a.s.}
$$
Since $Z=\left(Z_t\right)_{t\geq 0}$ is continuous and $Z_{\tau}=1$ $P$-a.s., this yields a contradiction, thus showing that $L=\left(L_t\right)_{t\geq 0}$ cannot be a local martingale deflator in $\GG$ on the global time horizon $\left[0,\tau\right]$.
\end{proof}

Proposition \ref{defl-2} represents a negative result, since it shows that there exists no local martingale deflator in $\GG$ on $[0,\infty]$. We close this section by showing that, if we restrict our attention to the time horizon $[\rho,\infty]$, for any $\GG$-stopping time $\rho$ with $\rho>\tau$ $P$-a.s., then there still exists a local martingale deflator in $\GG$. For every such $\GG$-stopping time $\rho$, let us introduce the process $\,^{\rho}\!S:=S-S^{\rho}=S_{\rho\vee\cdot}-S_{\rho}$. In the spirit of Proposition \ref{exdefl}, we have the following result.

\begin{Prop}	\label{exdefl-tau+eps}
For every $\GG$-stopping time $\rho$ such that $\rho>\tau$ $P$-a.s., the process $^{\rho}\!L=(\,^{\rho}\!L_t)_{t\geq 0}$ defined by, for all $t\geq0$,
$$
^{\rho}\!L_t := \frac{N^*_{\infty}-N_{\rho}}{N^*_{\infty}-N_{\rho\vee\, t}},
$$
is a local martingale deflator for the process $\,^{\rho}\!S$ with respect to the filtration $\GG$.
Furthermore, the process $\,^{\rho}\!L$ is a uniformly integrable $\GG$-martingale if and only if $P(\rho<\nu)=0$, with $\nu=\inf\{t\geq0:N_t=0\}$.
\end{Prop}
\begin{proof}
Observe first that the process $\,^{\rho}\!L$ is well-defined, since $N_t<N^*_{\infty}$ $P$-a.s. for every $t>\tau$. Moreover, $\,^{\rho}\!L$ is a $P$-a.s. strictly positive process with $\,^{\rho}\!L_{0}=1$ and, recalling that $\lim_{t\rightarrow\infty}N_t=0$ $P$-a.s., satisfies $\,^{\rho}\!L_{\infty}=(N^*_{\infty}-N_{\rho})/N^*_{\infty}>0$ $P$-a.s. To show that $\,^{\rho}\!L$ is a $\GG$-local martingale, we proceed by It\^o's formula as in the proof of Proposition \ref{exdefl}:
$$	\ba
d\,^{\rho}\!L_t &= \ind_{\{t>\rho\}}\frac{N^*_{\infty}-N_{\rho}}{\left(N^*_{\infty}-N_t\right)^2}\,dN_t
+\ind_{\{t>\rho\}}\frac{N^*_{\infty}-N_{\rho}}{\left(N^*_{\infty}-N_t\right)^3}\,
d\langle N\rangle_t	\\
&= \ind_{\{t>\rho\}}\frac{N^*_{\infty}-N_{\rho}}{\left(N^*_{\infty}-N_t\right)^2}\,\varphi_t
\left(dS_t+\frac{1}{N^*_{\infty}-N_t}\,d\langle S,N\rangle_t\right)
= \ind_{\{t>\rho\}}\frac{\,^{\rho}\!L_t}{N^*_{\infty}-N_t}\,\varphi_t\,
d\tildeS_t\,.
\ea $$
It remains to show that the product $\,^{\rho}\!S\,^{\rho}\!L$ is a $\GG$-local martingale:
$$	\ba
d\,^{\rho}\!S_t^{\rho}\!L_t &=
\,^{\rho}\!S_t\,d\,^{\rho}\!L_t + \,^{\rho}\!L_t\,d\,^{\rho}\!S_t + d\langle\,^{\rho}\!S,\,^{\rho}\!L\rangle_t	\\
&= \,^{\rho}\!S_t\,d\,^{\rho}\!L_t + \ind_{\{t>\rho\}}\,^{\rho}\!L_t\,dS_t + \ind_{\{t>\rho\}}\,^{\rho}\!L_t\frac{1}{N^*_{\infty}-N_t}\,d\langle S,N\rangle_t	\\
&= \,^{\rho}\!S_t\,d\,^{\rho}\!L_t + \ind_{\{t>\rho\}}\,^{\rho}\!L_t\,d\tildeS_t\,.
\ea	$$
The last claim of the proposition follows from the simple observation that $\,^{\rho}\!L$ is a uniformly integrable $\GG$-martingale if and only if $E[\,^{\rho}\!L_{\infty}]=1$. Since $\,^{\rho}\!L_{\infty}=1-N_{\rho}/N^*_{\infty}$, this holds if and only if $N_{\rho}=0$ $P$-a.s., i.e., if and only if $P(\rho<\nu)=0$.
\end{proof}

Note also that, since $N=0$ on $\lsi\nu,\infty\rsir$, the only case where the process $^{\rho}\!L$ is a uniformly integrable $\GG$-martingale is when it satisfies $^{\rho}\!L_t=1$ for all $t\geq0$ $P$-a.s.
In that case, the process $^{\rho}\!S$ is itself a $\GG$-local martingale.

\section{Arbitrages on the time horizon $[0,\tau]$}	\label{S4}

The goal of this section is to determine whether the information associated to an honest time $\tau$ does give rise to arbitrage profits in the enlarged financial market $\M^{\GG}$ on the time horizon $\left[0,\tau\right]$. Unless mentioned otherwise, we shall always suppose that Assumptions \ref{ass-NFLVR}, \ref{ass-avoid} and \ref{ass-PRP} are satisfied.

Recall first that Lemma \ref{NikYor} gives the existence of a non-negative $\FF$-local martingale $N$ with $N_0=1$ and $\lim_{t\rightarrow\infty}N_t=0$ $P$-a.s. such that $\tau=\sup\left\{t\geq 0:N_t=N^*_{\infty}\right\}$. It is clear that $N_{\tau}\geq 1$ $P$-a.s. as well as $P\left(N_{\tau}>1\right)>0$. Furthermore, due to Assumption \ref{ass-PRP}, there exists an $\R^d$-valued $\FF$-predictable process $\varphi=(\varphi_t)_{t\geq 0}\in L^{\FF}\left(S\right)$ such that $N=1+\varphi\cdot S$.
By relying on these arguments, we can easily construct an admissible $\GG$-strategy which yields an arbitrage opportunity (in the sense of part \emph{(ii)} of Definition \ref{def-arb}) in the enlarged financial market $\M^{\GG}$, as shown in part \emph{(ii)} of the following theorem. This gives a definite answer to the question of whether an agent in the enlarged financial market $\M^{\GG}$ can profit from the additional information and realise arbitrage profits on the time horizon $[0,\tau]$.

\begin{Thm}	\label{attau}
\mbox{}\\[-0.6cm]
\begin{itemize}
\item[(i)]
NA1 (or, equivalently, NUPBR) holds in $\M^{\GG}$ on the time horizon $\left[0,\tau\right]$;
\item[(ii)]
the strategy $\bar{\varphi}:=\ind_{\lsi0,\tau\rsi}\varphi$ yields an arbitrage opportunity in $\M^{\GG}$ and, hence, NA fails to hold in $\M^{\GG}$ on the time horizon $\left[0,\tau\right]$;
\item[(iii)]
NFLVR fails to hold in $\M^{\GG}$ on the time horizon $\left[0,\tau\right]$.
\end{itemize}
\end{Thm}
\begin{proof}
\emph{(i)}: this is a direct consequence of Proposition \ref{exdefl} together with part \emph{(i)} of Theorem \ref{char-arb}.\\
\emph{(ii)}: by Remark \ref{rem-filtr}, since $\tau$ is a $\GG$-stopping time and the process $\ind_{\lsi0,\tau\rsi}$ is left-continuous, it is clear that $\bar{\varphi}\in L^{\GG}\left(S\right)$. Since $V\left(0,\bar{\varphi}\right)_t=\left(\ind_{\lsi0,\tau\rsi}\varphi\cdot S\right)_t=N_{t\wedge\tau}-1\geq-1$ $P$-a.s. for all $t\geq 0$, we also have $\bar{\varphi}\in\A^{\GG}_1$. Note that $V\left(0,\bar{\varphi}\right)_{\infty}=V\left(0,\bar{\varphi}\right)_{\tau}=N_{\tau}-1$, thus implying $V\left(0,\bar{\varphi}\right)_{\tau}\geq 0$ $P$-a.s. and $P\bigl(V\left(0,\bar{\varphi}\right)_{\tau}>0\bigr)>0$. This shows that NA fails in the enlarged financial market $\M^{\GG}$ on the time horizon $\left[0,\tau\right]$.\\
\emph{(iii)}: this follows directly from part \emph{(iii)} of Theorem \ref{char-arb}.
\end{proof}

\begin{Rem}[An alternative proof of part (iii) of Theorem \ref{attau}]
It is worth noting that the failure of NFLVR in the enlarged financial market $\M^{\GG}$ on the time horizon $\left[0,\tau\right]$ can also be proved in a purely probabilistic way, by relying on the properties of local martingale deflators in $\GG$ on $[0,\tau]$ established in Section \ref{S3.1}. Indeed, suppose on the contrary that NFLVR holds on $\left[0,\tau\right]$. In view of part \emph{(ii)} of Theorem \ref{char-arb}, this implies the existence of a probability measure $Q\sim P$, with density process $L_t:=\frac{dQ|_{\G_t}}{dP|_{\G_t}}$, $t\geq 0$, such that $S^{\tau}$ is a $\left(Q,\GG\right)$-local martingale. Obviously, the process $L/L_0=\left(L_t/L_0\right)_{t\geq 0}$ is a local martingale deflator in $\GG$ on the time horizon $\left[0,\tau\right]$ and also a uniformly integrable $\GG$-martingale. However, this contradicts the last part of Lemma \ref{str-defl} and, hence, NFLVR cannot hold in the enlarged financial market $\M^{\GG}$ on $\left[0,\tau\right]$.
\end{Rem}

In particular, despite its simplicity, the result of part \emph{(ii)} of Theorem \ref{attau} is quite interesting. Indeed, it shows that, as long as Assumptions \ref{ass-NFLVR}, \ref{ass-avoid} and \ref{ass-PRP} hold, one can explicitly construct an admissible $\GG$-strategy which realises an arbitrage opportunity \emph{at} the honest time $\tau$. To the best of our knowledge, this result is new: as mentioned in the introduction, all previous works in the literature have only shown the existence of arbitrage opportunities \emph{immediately after} $\tau$ (see e.g. \citet{Imk} and \citet{Zw}).

\begin{Rem}
The arbitrage strategy $\bar{\varphi}$ constructed in part \emph{(ii)} of Theorem \ref{attau} admits a special interpretation. Indeed, the corresponding value process $V\left(1,\bar{\varphi}\right)=N^{\tau}$ is the reciprocal of the local martingale deflator $1/N^{\tau}$ (see Proposition \ref{exdefl}). According to Theorem 7 of \citet{HS} (see also \citet{KK}, Section 4.4), this means that $V\left(1,\bar{\varphi}\right)$ represents the value process of the \emph{growth-optimal portfolio}, which also coincides with the \emph{numéraire portfolio}, for the enlarged financial market $\M^{\GG}$ on the time horizon $\left[0,\tau\right]$.
\end{Rem}

\begin{Ex}[Discussion of Example \ref{example}]	\label{rem-BS}
Let $d=1$ and suppose that the real-valued process $S=\left(S_t\right)_{t\geq 0}$ is given as the solution of the SDE \eqref{BS} on the filtered probability space $\left(\Omega,\F,\FF^W\!,P\right)$, where $\FF^W$ is the ($P$-augmented) natural filtration of $W$. Since $S$ is a $\left(P,\FF\right)$-martingale, Assumption \ref{ass-NFLVR} is trivially satisfied and, clearly, Assumption \ref{ass-PRP} holds as well. Furthermore, since $\lim_{t\rightarrow\infty}S_t=0$ $P$-a.s. (due to the law of large numbers for Brownian motion), Corollary 2.4 of \citet{NY} implies that the random time $\tau=\sup\left\{t\geq 0:S_t=\sup_{u\geq 0}S_u\right\}$ is an honest time which avoids all $\FF$-stopping times, meaning that Assumption \ref{ass-avoid} is also satisfied. Note that, in the context of this example, we have $S/S_0=N$, as can be deduced from Proposition 2.2 of \citet{NY}, and $\nu=\infty$ $P$-a.s. Then, Theorem \ref{attau} directly imply claim (a) after Example \ref{example}. Observe, furthermore, that, in the context of this simple example, the arbitrage opportunity constructed in part \emph{(ii)} of Theorem \ref{attau} reduces to a simple buy-and-hold position on the risky asset $S$ until time $\tau$.
\end{Ex}

We close this section with two examples, which in particular show how the arbitrage strategy $\bar{\varphi}$ appearing in part \emph{(ii)} Theorem \ref{attau} can be explicitly calculated. The first example below is based on the \emph{last passage time} of a geometric Brownian motion, while the second example is linked to the \emph{drawdown} of a geometric Brownian motion.

\begin{Ex}[An arbitrage opportunity arising from a last passage time]	\label{last passage}
As in Example \ref{example}, let the discounted price process $S=(S_t)_{t\geq 0}$ of a risky asset be modeled as in \eqref{BS} on the filtered probability space $(\Omega,\F,\FF^{W}\!,P)$ and define the random time $\tau:=\sup\{t\geq 0:S_t=a\}$, for some $a\in(0,S_0)$. Since $\lim_{t\rightarrow\infty}S_t=0$ $P$-a.s., the random time $\tau$ is easily seen to be a $P$-a.s. finite honest time. Moreover, following Section 5.6 of \citet{JYC}, the associated Azéma supermartingale $Z=(Z_t)_{t\geq 0}$ is given by $Z_t=(S_t/a)\wedge 1$, for all $t\geq 0$. By Tanaka's formula (see e.g. \citet{JYC}, Section 4.1.8), the Doob-Meyer decomposition of $Z$ can be computed as:
$$
Z_t = 1+\frac{1}{a}\int_0^t\!\ind_{\{S_u<a\}}dS_u-\frac{1}{2a}L^a_t,
\qquad \text{for all }t\geq 0,
$$
where $L^a=(L^a_t)_{t\geq 0}$ denotes the local time of $S$ at the level $a$. Note that, in view of Remark 1.2 in \citet{MY} together with the continuity of $L^a$, the honest time $\tau$ satisfies Assumption \ref{ass-avoid}.
Since $Z_t>0$ $P$-a.s. for all $t\geq 0$, Theorem II.8.21 of \citet{JacShi} implies that the Azéma supermartingale $Z$ admits a multiplicative decomposition of the form $Z=N/D$, where:
$$
D = \exp\left(\frac{L^a}{2a}\right)
\qquad\text{and}\qquad
N = \mathcal{E}\!\left(\int\!\frac{1}{aZ}\ind_{\{S<a\}}dS\right)
= 1+\int\!\frac{D}{a}\ind_{\{S<a\}}dS
$$
where we have used the fact that, for almost all $\omega\in\Omega$, the measure $dL^a_{\cdot}(\omega)$ is supported by the set $\{t\geq 0:S_t(\omega)=a\}=\{t\geq 0:Z_t(\omega)=1\}$, due to Theorem IV.69 of \citet{Pr}.
Since $\lim_{t\rightarrow\infty}Z_t=0$ $P$-a.s., the continuous $\FF$-local martingale $N$ satisfies $\lim_{t\rightarrow\infty}N_t=0$ $P$-a.s. Furthermore, Skorohod's reflection lemma (see \citet{JYC}, Lemma 4.1.7.1) and the arguments used in the proof of Theorem 4.1 of \citet{NY} allow to check that $D=N^*$, with $N^*$ denoting the running supremum of $N$. This gives a complete and explicit description of the multiplicative decomposition $Z=N/N^*$ appearing in Lemma \ref{NikYor}. Since $\tau=\sup\{t\geq 0: Z_t=1\}=\sup\{t\geq 0: N_t=N^*_{\infty}\}$, the strategy $\bar{\varphi}:=\ind_{\lsi0,\tau\rsi}\frac{1}{a}\exp\bigl(\frac{L^a}{2a}\bigr)\ind_{\{S<a\}}\in\A_1^{\GG}$ yields an arbitrage opportunity in the enlarged financial market $\M^{\GG}$, as shown in part \emph{(ii)} of Theorem \ref{attau}.

In the context of the present example, it is interesting to remark that the natural candidate $\bar{\varphi}:=-\ind_{\lsi0,\tau\rsi}$ (i.e., a short position on the risky asset $S$ until time $\tau$) for an arbitrage strategy fails to be an admissible strategy, since the process $-S$ is unbounded from below.
However, an alternative buy-and-hold admissible arbitrage strategy can be constructed as follows\footnote{We are thankful to an associate editor for having pointed out to us this alternative arbitrage strategy.}: for $b\in(0,a)$, let us define the stopping time $\tau_b:=\inf\{t\geq0:S_t=b\}$ and the $\GG$-predictable process $\bar{\varphi}':=\ind_{\lsir\tau_b\wedge\tau,\tau\rsi}$. Since $V(0,\bar{\varphi}')_t=\ind_{\{\tau_b<\,t\wedge\tau\}}\left(S_{t\wedge\tau}-S_{\tau_b}\right)\geq -S_{\tau_b}=-b$ $P$-a.s., for all $t\geq 0$, we have $\bar{\varphi}'\in\A^{\GG}_b$. Moreover, $V(0,\bar{\varphi}')_{\infty}=\ind_{\{\tau_b<\tau\}}\left(S_{\tau}-b\right)=\ind_{\{\tau_b<\tau\}}\left(a-b\right)\geq0$ $P$-a.s. and
Doob's maximal identity (see e.g. \citet{MY}, Lemma 0.1) implies that
$$
P(\tau_b<\tau) = 1-P\biggl(\sup_{t\geq\tau_b}S_t<a\biggr)
= 1-P\biggl(\frac{S_{\tau_b}}{\sup_{t\geq\tau_b}S_t}>\frac{b}{a}\biggr) = \frac{b}{a}>0\,,
$$
thus proving that $\bar{\varphi}'$ realises an arbitrage opportunity in $\M^{\GG}$ on $[0,\tau]$.
\end{Ex}

\begin{Ex}[An arbitrage opportunity arising from a random time related to the drawdown]	\label{drawdown}
As in the preceding example, let the discounted price process $S=(S_t)_{t\geq 0}$ of a risky asset be modeled as in \eqref{BS} on the filtered probability space $(\Omega,\F,\FF^W\!,P)$. The \emph{drawdown} process $\bar{D}=(\bar{D}_t)_{t\geq 0}$ is defined as $\bar{D}_t:=S^*_t-S_t$, for all $t\geq 0$, with $S^*$ denoting the running supremum of $S$. For a fixed $K\in(0,S_0)$, define the $\FF$-stopping time $\tau_K:=\inf\{t\geq 0:\bar{D}_t=K\}$ and the random time $\tau:=\sup\{t\leq\tau_K:S_t=S^*_t\}$. Since $\lim_{t\rightarrow\infty}S_t=0$ $P$-a.s., it is clear that  $\tau$ is a $P$-a.s. finite honest time.
Recalling that the scale function of the diffusion \eqref{BS} can be chosen to be the identity function, Proposition 1 of \citet{ZH} implies that the Azéma supermartingale associated to $\tau$ is given by $Z=\ind_{\lsi0,\tau_K\rsir}(K-\bar{D})/K$ and admits the following Doob-Meyer decomposition, for all $t\geq 0$:
$$
Z_t = 1 + \frac{\bar{\sigma}}{K}\int_0^{t\wedge\tau_K}\!S_u\,dW_u - \frac{1}{K}\left(S^*_{t\wedge\tau_K}-S_0\right).
$$
Note that, in view of Remark 1.2 in \citet{MY}, the honest time $\tau$ satisfies Assumption \ref{ass-avoid}. Furthermore, Theorem II.8.21 of \citet{JacShi} implies that the supermartingale $Z$ admits a multiplicative decomposition of the form $Z=\ind_{\lsi0,\tau_K\rsir}\,N/D$, where:
$$
D = \exp\left(\frac{S^*_{\cdot\wedge\tau_K}-S_0}{K}\right)
\quad\text{and}\quad
N = \mathcal{E}\!\left(\frac{\bar{\sigma}}{K}\int\!\!\ind_{\lsi0,\tau_K\rsir}\frac{S}{Z}\,dW\right)
= 1+\int\!\!\ind_{\lsi0,\tau_K\rsir}\frac{D}{K}\,dS
$$
where we have used the fact that, for almost all $\omega\in\Omega$, the measure $dS_{\cdot}^*(\omega)$ is supported by the set $\{t\geq 0:S_t(\omega)=S^*_t(\omega)\}=\{t\geq 0:Z_t(\omega)=1\}$. 
Analogously to Example \ref{last passage}, it can be checked that $D=N^*$ on $\lsi0,\tau_K\rsir$, meaning that $Z=\ind_{\lsi0,\tau_K\rsir}\,N/N^*$. 
The fact that $\lim_{t\rightarrow\tau_K}Z_t=0$ $P$-a.s. implies that $\lim_{t\rightarrow\tau_K}N_t=0$ $P$-a.s. and, hence, due to the minimum principle for non-negative supermartingales (see \citet{JacShi}, Lemma III.3.6), we can write $Z=N/N^*$.
Since $\tau=\sup\{t\geq 0:Z_t=1\}=\sup\{t\geq 0:N_t=N^*_{\infty}\}$, the strategy $\bar{\varphi}:=\ind_{\lsi0,\tau\rsi}D/K\in\A^{\GG}_1$ yields an arbitrage opportunity in the enlarged financial market $\M^{\GG}$. By relying on \citet{ZH}, the present example can be easily generalised to the case where the constant parameter $\bar{\sigma}$ in \eqref{BS} is replaced by a strictly positive continuous function $\bar{\sigma}(\cdot)$ evaluated at $S_t$ (compare also with \citet{CNP}, Section 5).
\end{Ex}

\section{Arbitrages on the time horizon $[0,\sigma\wedge\tau]$}	\label{S5}

As shown in the preceding section, even though arbitrages of the first kind can never be realised in the enlarged financial market $\M^{\GG}$ on the time horizon $[0,\tau]$, one can profit from arbitrage opportunities at $\tau$. In the present section, we study whether it is possible to exploit the information of the progressively enlarged filtration $\GG$ in order to obtain arbitrage opportunities \emph{before} the honest time $\tau$. The answer to such a question is given by the following theorem, which relies on Lemma \ref{comp}.

\begin{Thm}	\label{beforetau}
Let $\sigma$ be an $\FF$-stopping time. Then NFLVR holds in the enlarged financial market $\M^{\GG}$ on the time horizon $\left[0,\sigma\wedge\tau\right]$ if and only if $P\left(\sigma\geq\nu\right)=0$.
\end{Thm}
\begin{proof}
If $P\left(\sigma\geq\nu\right)=0$, equation \eqref{str-defl-sigma} (with $k\!=\!\eta\!=\!0$) together with Lemma \ref{comp} implies that the process $1/N^{\sigma\wedge\tau}$ is a uniformly integrable $\GG$-martingale. Together with Proposition \ref{exdefl}, this shows that $1/N^{\sigma\wedge\tau}$ can be taken as the density process of an ELMM$_{\GG}$ for $S^{\sigma\wedge\tau}$. Due to part \emph{(ii)} of Theorem \ref{char-arb}, it follows that NFLVR holds in the enlarged financial market $\M^{\GG}$ on the time horizon $\left[0,\sigma\wedge\tau\right]$. Conversely, if $P\left(\sigma\geq\nu\right)>0$, Lemma \ref{comp} implies that $E\left[L_{\sigma\wedge\tau}\right]<1$ for any local martingale deflator $L=\left(L_t\right)_{t\geq 0}$ in $\GG$ on the time horizon $\left[0,\sigma\wedge\tau\right]$. This implies that $L^{\sigma\wedge\tau}$ cannot be a uniformly integrable martingale and, hence, no ELMM$_{\GG}$ can exist for the enlarged financial market $\M^{\GG}$ on the time horizon $[0,\sigma\wedge\tau]$.
\end{proof}

\begin{Rems}
1) Note that, for any $\FF$-stopping time $\sigma$ with $P\left(\sigma\geq\nu\right)=0$, we always have $P\left(\sigma<\tau\right)=E\left[Z_{\sigma}\right]=E\left[Z_{\sigma}\ind_{\left\{\sigma<\nu\right\}}\right]>0$. Hence, there is no contradiction between Theorem \ref{beforetau} and Theorem \ref{attau}.
Furthermore, Theorem XX.14 of \citet{DMM} shows that $\nu$ is the smallest $\FF$-stopping time which is $P$-a.s. greater than $\tau$.

2) Theorem \ref{beforetau} admits an immediate generalisation to $\GG$-stopping times. Indeed, due to the lemma on page 370 of Chapter VI of \citet{Pr}, for every $\GG$-stopping time $\sigma$ there exists an $\FF$-stopping time $\sigma'$ such that $\sigma\wedge\tau=\sigma'\wedge\tau$ $P$-a.s. Then, it can be readily checked that NFLVR holds in the enlarged financial market $\M^{\GG}$ on the time horizon $[0,\sigma\wedge\tau]$ if and only if $P(\sigma'\geq\nu)=0$.
\end{Rems}

In particular, due to Theorem \ref{beforetau}, the NFLVR condition holds in the enlarged financial market $\M^{\GG}$ on the time horizon $\left[0,\tau\wedge T\right]$ for any $T\in\left(0,\infty\right)$ satisfying $P\left(\nu\leq T\right)=0$. We also have the following corollary, which shows that one can never obtain arbitrage opportunities in the enlarged financial market $\M^{\GG}$ \emph{strictly before} the honest time $\tau$.

\begin{Cor}	\label{cor-beforetau}
Let $\varrho$ be a $\GG$-stopping time with $\varrho<\tau$ $P$-a.s. Then NFLVR holds in the enlarged financial market $\M^{\GG}$ on the time horizon $\left[0,\varrho\right]$.
\end{Cor}
\begin{proof}
If $\varrho$ is a $\GG$-stopping time with $\varrho<\tau$, the lemma on page 370 of \citet{Pr} implies that there exists an $\FF$-stopping time $\sigma$ with $\sigma=\varrho$ $P$-a.s. Noting that $\tau<\nu$ $P$-a.s. (see the beginning of the proof of Proposition \ref{exdefl}), the claim then follows from Theorem \ref{beforetau}.
\end{proof}

\begin{Rem}
Actually, the result of Corollary \ref{cor-beforetau} holds true in a general semimartingale setting, as long as the restricted financial market $\M^{\FF}$ satisfies NFLVR, regardless of the validity of Assumptions \ref{ass-avoid}-\ref{ass-PRP} and of the fact that $\tau$ is an honest time. Indeed, let $\tau$ be an arbitrary random time and let $\varrho$ be a $\GG$-stopping time with $\varrho<\tau$ $P$-a.s. As in the proof of Corollary \ref{cor-beforetau}, there exists an $\FF$-stopping time $\sigma$ with $\sigma=\rho$ $P$-a.s. On the time horizon $[0,\varrho]$, every admissible strategy $\theta\in\A^{\GG}$ can be assumed to satisfy $\theta=\ind_{\lsi0,\varrho\rsi}\theta$. Due to Lemma 4.4 of \citet{Jeu}, there always exists an $\FF$-predictable process $\psi$ such that $\theta=\ind_{\lsi0,\varrho\rsi}\psi=\ind_{\lsi0,\sigma\rsi}\psi$. Furthermore, the stochastic integrals $(\theta\cdot S)^{\varrho}=\left(\psi\cdot S\right)^{\sigma}$ are indistinguishable and the stochastic integral $\left(\psi\cdot S\right)^{\sigma}$ coincides in the two filtrations $\FF$ and $\GG$ (see \citet{JacShi}, Proposition III.6.25). This implies that, on the time horizon $[0,\varrho]$, every outcome of an admissible $\GG$-strategy can also be realised as the outcome of an admissible $\FF$-strategy. As a consequence, if $\M^{\FF}$ satisfies NFLVR, then NFLVR also holds for $\M^{\GG}$ on $[0,\varrho]$.
\end{Rem}

Note that, in the special context considered in Example \ref{example} (see also Example \ref{rem-BS}), Theorem \ref{beforetau} and Corollary \ref{cor-beforetau} together imply claim (b) after Example \ref{example}.

\begin{Ex} 
Due to Theorem \ref{beforetau}, if $\sigma$ is an $\FF$-stopping time such that $P\left(\sigma\geq\nu\right)>0$, then there exist arbitrage opportunities in the enlarged financial market $\M^{\GG}$ on the time horizon $\left[0,\sigma\wedge\tau\right]$. Let us illustrate this fact by means of a simple example, in the same setting of Example \ref{example}. Suppose that $S_0=1$ and let us define the $\FF$-stopping time $\tau^*:=\inf\bigl\{t\geq 0:S_t=1/2\bigr\}$, which is $P$-a.s. finite, and the honest time $\tau:=\sup\bigl\{t\in\left[0,\tau^*\right]:S_t=S^*_t\bigr\}$. Let us also introduce the $\FF$-stopping time $\sigma:=\inf\left\{t\geq 0:S_t=3/2\right\}$. It can be checked that $\tau$ avoids all $\FF$-stopping times and we have $\nu=\tau^*$ $P$-a.s. and $P\left(\sigma>\nu\right)=P\left(\sigma>\tau^*\right)>0$. Hence, due to Theorem \ref{beforetau}, NFLVR fails to hold in the enlarged financial market $\M^{\GG}$ on $\left[0,\sigma\wedge\tau\right]$. Indeed, the buy-and-hold strategy $\ind_{\lsi0,\sigma\wedge\tau\rsi}$ provides an arbitrage opportunity on the time horizon $\left[0,\sigma\wedge\tau\right]$, since $S_{\sigma\wedge\tau}-S_0\geq 0$ $P$-a.s. and $P\left(S_{\sigma\wedge\tau}>S_0\right)>0$.
In general, it remains an open problem how to explicitly construct arbitrage strategies based on the local martingale $N$, in the spirit of Theorem \ref{attau}, when $P(\sigma\geq\nu)>0$.
\end{Ex}

\section{Arbitrages on the time horizon $[0,\infty]$}	\label{S6}

In this section, we study the existence of arbitrage profits in the enlarged financial market $\M^{\GG}$ on the global time horizon $\left[0,\infty\right]$, taking into account especially what can happen \emph{after} the honest time $\tau$. 
Note that we already know from Theorem \ref{attau} that NA and NFLVR fail to hold in the enlarged financial market $\M^{\GG}$ on $\left[0,\tau\right]$ and, as a consequence, also on $\left[0,\infty\right]$. Hence, on the time horizon $[0,\infty]$, we shall restrict our attention to weaker no-arbitrage-type conditions, notably the NA1 condition.
We shall always suppose that Assumptions \ref{ass-NFLVR}, \ref{ass-avoid} and \ref{ass-PRP} are satisfied and that, without loss of generality, $P$ is already an ELMM$_{\FF}$ for $S$ (see the beginning of Section \ref{S3}).

Before analysing the validity of NA1 in the enlarged financial market $\M^{\GG}$ on the time horizon $[0,\infty]$, let us introduce a particularly strong notion of arbitrage profit.

\begin{Def}	\label{def-NUIP}
An element $\theta\in\A^{\GG}_0$ yields an \emph{Unbounded Increasing Profit} if
$$
P\bigl(V\left(0,\theta\right)_s\leq V\left(0,\theta\right)_t,\text{ for all }0\leq s \leq t\leq\infty\bigr)=1
\quad\text{and}\quad
P\bigl(V\left(0,\theta\right)_{\infty}>0\bigr)>0\,.
$$
If there exists no such $\theta\in\A^{\GG}_0$ we say that the financial market $\M^{\GG}$ satisfies the \emph{No Unbounded Increasing Profit (NUIP)} condition.
\end{Def}

The notion of \emph{Unbounded Increasing Profit} has been introduced under that name in \citet{KK} and is stronger than all the notions of arbitrage given in Definition \ref{def-arb}. Indeed, it can be directly checked that the existence of an unbounded increasing profit implies that none of the NA1, NA and NFLVR conditions can hold.
The next simple lemma shows that Unbounded Increasing Profits can never be realised in the enlarged financial market $\M^{\GG}$.

\begin{Lem}	\label{NUIP}
NUIP holds in the enlarged financial market $\M^{\GG}$ on the global time horizon $[0,\infty]$.
\end{Lem}
\begin{proof}
Suppose that $\theta\in\A^{\GG}_0$ generates an unbounded increasing profit. Due to Definition \ref{def-NUIP}, the process $V\left(0,\theta\right)=\theta\cdot S$ is increasing, hence of finite variation. According to the notation introduced in \eqref{candec-2}, this implies that the $\GG$-local martingale $\theta\cdot\tildeS=\theta\cdot S-\theta\cdot\tildeA$ is null, being $\GG$-predictable and of finite variation, so that $\bigl|\bigl\langle\theta\cdot\tildeS,\tildeS^i\bigr\rangle\bigr|=0$ $P$-a.s. for all $i=1,\ldots,d$. It then follows, for all $t\geq 0$:
$$
V\left(0,\theta\right)_t = \left(\theta\cdot S\right)_t = (\theta\cdot\tildeA)_t
= \int_0^t\theta_s'\,d\bigl\langle\tildeS,\tildeS\,\bigr\rangle_{\!s}\,\tilde{\alpha}_{s}
= \int_0^td\bigl\langle\theta\cdot\tildeS,\tildeS\,\bigr\rangle_{\!s}\,\tilde{\alpha}_{s} = 0
\qquad \text{$P$-a.s.}
$$
thus contradicting the assumption that $P\bigl(V\left(0,\theta\right)_{\infty}>0\bigr)>0$.
\end{proof}

However, Lemma \ref{NUIP} only excludes the existence of an almost pathological notion of arbitrage.
The question of whether the more meaningful NA1 condition holds in the enlarged financial market $\M^{\GG}$ on the time horizon $[0,\infty]$ is negatively answered by the following proposition, which explicitly exhibits an arbitrage of the first kind, in the sense of part \emph{(i)} of Definition \ref{def-arb}.

\begin{Prop}	\label{arb2}
The random variable $\xi:=N_{\tau}-1$ yields an arbitrage of the first kind. As a consequence, NA1 (or, equivalently, NUPBR) fails to hold in the enlarged financial market $\M^{\GG}$ on the global time horizon $\left[0,\infty\right]$.
\end{Prop}
\begin{proof}
Clearly, $\xi:=N_{\tau}-1$ is an $\F$-measurable non-negative random variable with $P\left(\xi>0\right)>0$.
Let $\hat{\varphi}:=-\ind_{\lsir\tau,\infty\rsir}\varphi$. Due to Remark \ref{rem-filtr} and since $\tau$ is a $\GG$-stopping time, it is clear that $\hat{\varphi}\in L^{\GG}\left(S\right)$. Since $V\left(0,\hat{\varphi}\right)=-\ind_{\lsir\tau,\infty\rsir}\varphi\cdot S=N^{\tau}-N$, we have $V\left(0,\hat{\varphi}\right)_t=0$ on $\left\{t\leq\tau\right\}$ and also $V\left(0,\hat{\varphi}\right)_t>0$ on $\left\{t>\tau\right\}$, because of the fact that $\tau=\sup\left\{t\geq 0:N_t=N^*_{\infty}\right\}$ (see Lemma \ref{NikYor}), thus implying that $\hat{\varphi}\in\A_0^{\GG}$. For all $x>0$, we then have $V\left(x,\hat{\varphi}\right)_{\infty}=x+N_{\tau}-N_{\infty}=x+1+\xi>\xi$, thus showing that $\xi$ yields an arbitrage of the first kind.
\end{proof}

\begin{Rem}[An alternative proof of the failure of NA1]	\label{proof-arb2}
We have chosen to present a constructive and simple proof of Proposition \ref{arb2}, by explicitly exhibiting an arbitrage of the first kind. However, the failure of NA1 in the enlarged financial market $\M^{\GG}$ on the time horizon $[0,\infty]$ can also be proved by relying on purely probabilistic arguments. Indeed, Proposition \ref{defl-2} shows that there exists no local martingale deflator in $\GG$ on the time horizon $[0,\infty]$. Together with part \emph{(i)} of Theorem \ref{char-arb}, this implies that NA1 fails to hold in $\M^{\GG}$ on $[0,\infty]$.
\end{Rem}

It has already been shown in \citet{Imk} and \citet{Zw} that NFLVR fails to hold \emph{after} $\tau$ in the enlarged financial market $\M^{\GG}$. However, the proofs given in those papers are somehow abstract and technical. In contrast, the proof of Proposition \ref{arb2} is extremely simple and explicitly shows the trading strategy which realises the arbitrage profit. Furthermore, we have shown that not only NA and NFLVR but also the weaker NA1 and NUPBR no-arbitrage-type conditions fail to hold in the enlarged financial market $\M^{\GG}$ on the global time horizon $\left[0,\infty\right]$.

\begin{Rems}	\label{rem-BS-3}
1) As shown in the proof of Proposition \ref{arb2}, the strategy $\hat{\varphi}\in\A^{\GG}_0$ satisfies $\hat{\varphi}=\hat{\varphi}\ind_{\lsir\tau,\infty\rsir}$ and $\left(\hat{\varphi}\cdot S\right)_t>0$ for all $t>\tau$. According to Definition 3.2 of \citet{DS2}, the strategy $\hat{\varphi}$ generates an \emph{immediate arbitrage opportunity} at the $\GG$-stopping time $\tau$. This intuitively means that one can realise an arbitrage profit \emph{immediately after} the $\GG$-stopping time $\tau$ has occurred, i.e., on the time interval $\left[\tau,\tau+\varepsilon\right]$, for every $\varepsilon>0$. This possibility has been also pointed out in \citet{Zw}.

2) Proposition \ref{arb2} can be seen as a counterpart to part \emph{(ii)} of Theorem \ref{attau}. Indeed, part \emph{(ii)} of Theorem \ref{attau} shows that one can realise an arbitrage opportunity \emph{at} time $\tau$ by taking a position (up to $\tau$) in the strategy $\varphi$ which replicates $N$, while Proposition \ref{arb2} shows that one can realise an arbitrage of the first kind (as well as an immediate arbitrage opportunity) \emph{after} time $\tau$ by taking a position in the strategy $-\varphi$. It is interesting to observe that in both cases the arbitrage strategy is directly related to the $\FF$-local martingale $N$ appearing in the multiplicative decomposition of the Azéma supermartingale $Z$ of $\tau$. Note also that the admissibility constraint prevents the arbitrage of the first kind $\xi=N_{\tau}-1$ to be realised on the time horizon $[0,\tau]$.

3) As considered in Example \ref{example}, let $d=1$ and suppose that the real-valued process $S=\left(S_t\right)_{t\geq 0}$ is given as the solution of the SDE \eqref{BS} on the filtered probability space $(\Omega,\F,\FF^W\!,P)$, where $\FF^W$ is the ($P$-augmented) natural filtration of $W$. In view of Example \ref{rem-BS}, the random time $\tau$ is an honest time which avoids every $\FF$-stopping time. Hence, Proposition \ref{arb2} directly implies claim (c) after Example \ref{example}. Observe that, in this simple example, the arbitrage strategy $\hat{\varphi}$ constructed in the proof of Proposition \ref{arb2} reduces simply to a short position on $S$ from time $\tau$ onwards.
\end{Rems}

As shown in Proposition \ref{arb2}, an agent can realise an arbitrage of the first kind in the enlarged financial market $\M^{\GG}$ by adopting a suitable trading strategy \emph{as soon as} the honest time $\tau$ occurs.
Motivated by this result, let us now consider what happens in the enlarged financial market $\M^{\GG}$ if all market participants are only allowed to trade on the time horizon $[\rho,\infty]$, where $\rho$ is a $\GG$-stopping time with $\rho>\tau$ $P$-a.s.  
To this effect, let us recall the process $^{\rho}\!S=S-S^{\rho}$ introduced at the end of Section \ref{S3.3} and define the \emph{$\rho$-shifted enlarged financial market} $\,^{\rho}\!\M^{\GG}:=\left\{\Omega,\F,\GG,P;\,^{\rho}\!S,\,^{\rho}\!\A^{\GG}\right\}$, where $^{\rho}\!\A^{\GG}$ denotes the set of all elements in $L^{\GG}(^{\rho}\!S)$ which are admissible, in the sense of Definition \ref{def-adm}.

\begin{Prop}	\label{arb-tau+eps}	\mbox{}
\begin{itemize}
\item[(i)]
For every $\GG$-stopping time $\rho$ such that $\rho>\tau$ $P$-a.s., NA1 holds in the $\rho$-shifted enlarged financial market $\,^{\rho}\!\M^{\GG}$;
\item[(ii)]
there exists a constant $\delta>0$ such that the strategy $\check{\varphi}:=-\ind_{\lsir\tau,\infty\rsir}\varphi/N^*_{\tau}$ yields an arbitrage opportunity in the $(\tau+\eps)$-shifted enlarged financial market $\,^{\tau+\eps}\!\M^{\GG}$, for every $\eps\in(0,\delta)$, and, hence, NA and NFLVR fail to hold in the $(\tau+\eps)$-shifted enlarged financial market $\,^{\tau+\eps}\!\M^{\GG}$, for every $\eps\in(0,\delta)$.
\end{itemize}
\end{Prop}
\begin{proof}
\emph{(i)}: this is a direct consequence of Proposition \ref{exdefl-tau+eps} together with part \emph{(i)} of Theorem \ref{char-arb}.\\
\emph{(ii)}: since $P(\tau<\nu)=1$, where $\nu=\inf\{t\geq0:N_t=0\}$ (compare the proof of Proposition \ref{exdefl}), there exists a constant $\delta>0$ such that $P(\nu>\tau+\eps)>0$ for every $\eps\in(0,\delta)$. Since $\varphi\in L^{\FF}(S)\subset L^{\GG}(S)$ and $\tau$ is a $\GG$-stopping time, it is clear that $\check{\varphi}\in L^{\GG}(^{\tau+\eps}\!S)$, for every $\eps\in(0,\delta)$. Furthermore, for all $t\geq0$:
$$
V(0,\check{\varphi})_t = \left(\check{\varphi}\cdot\,^{\tau+\eps}\!S\right)_t
= -\frac{N_{(\tau+\eps)\vee\, t}-N_{\tau+\eps}}{N^*_{\tau}}
\geq -\frac{N_{(\tau+\eps)\vee\, t}}{N^*_{\infty}} >-1
\qquad \text{$P$-a.s.}
$$
thus showing that $\check{\varphi}\in\,^{\tau+\eps}\!\A_1^{\GG}$.
Moreover, we have that $V(0,\check{\varphi})_{\infty}=N_{\tau+\eps}/N^*_{\infty}\geq0$ $P$-a.s. and $P\bigl(V(0,\check{\varphi})_{\infty}>0\bigr)=P(\nu>\tau+\eps)>0$, meaning that the strategy $\check{\varphi}$ yields an arbitrage opportunity in the $(\tau+\eps)$-shifted enlarged financial market $\,^{\tau+\eps}\!\M^{\GG}$.
\end{proof}

In particular, Proposition \ref{arb2} together with part \emph{(i)} of the above proposition shows that the possibility of realising arbitrages of the first kind in the enlarged financial market $\M^{\GG}$ crucially depends on the possibility of trading as soon as $\tau$ has occurred. Indeed, from time $\tau+\eps$ onwards, the additional knowledge of the information of the filtration $\GG$ can only allow for arbitrage opportunities, since NA1 holds in the $(\tau+\eps)$-shifted enlarged financial market $\,^{\tau+\eps}\!\M^{\GG}$. In other words, the potential loss of the viability (in the sense of \citet{Kar}) of the enlarged financial market $\M^{\GG}$ on the time horizon $[0,\infty]$ is only due to the arbitrage profits that can be realised by trading as soon as the honest time $\tau$ occurs. Hence, preventing agents from trading on the time horizon $(\tau,\tau+\eps)$, for every $\eps>0$, will preserve the viability of the enlarged financial market $\M^{\GG}$, even though arbitrage opportunities may still exist.
Note also that, in the special context considered in Example \ref{example} (see also Example \ref{rem-BS}), Proposition \ref{arb2} and Proposition \ref{arb-tau+eps} together imply claim (c) after Example \ref{example}.

We close this section by showing that the completeness of the restricted financial market $\M^{\FF}$ (see Assumption \ref{ass-PRP} and Remark \ref{complete}) or, more specifically, the existence of a stochastic integral representation of the form $N=1+\varphi\cdot S$, for some $\varphi\in L^{\FF}(S)$, is crucial for our results to hold (in particular, see the proofs of part \emph{(ii)} of Theorem \ref{attau}, Proposition \ref{arb2} and part \emph{(ii)} of Proposition \ref{arb-tau+eps}), as we are now going to illustrate by means of an explicit counterexample.

\begin{Ex}	\label{rem-counterexample}
Let $W^1=(W^1_t)_{t\geq 0}$ and $W^2=(W^2_t)_{t\geq 0}$ be two independent Brownian motions and denote by $\FF^i=(\F^i_t)_{t\geq 0}$ the $P$-augmented natural filtration of $W^i$, for $i=1,2$. Define $\FF:=\FF^1\vee\FF^2$ and let the discounted price process $S=(S_t)_{t\geq0}$ of a risky asset be given as the solution to the following SDE on the filtered probability space $(\Omega,\F,\FF,P)$:
\be	\label{SDE-counter}	\left\{	\ba
dS_t &= S_t\,f(W^1_t)\,dW^2_t,	\\
S_0 &= s\in\left(0,\infty\right),
\ea	\right. \ee
where $f:\R\rightarrow(0,\infty)$ is such that the above SDE admits a unique strong solution. Clearly, since $S$ is an $\FF$-local martingale, NFLVR holds in the financial market $\M^{\FF}$ (Assumption \ref{ass-NFLVR}). 
Let $\tau$ be any $P$-a.s. finite honest time with respect to the filtration $\FF^1$ and denote by $\GG^1=(\G^1_t)_{t\geq0}$ and $\GG=(\G_t)_{t\geq0}$ the progressive enlargements of $\FF^1$ and $\FF$, respectively, with respect to $\tau$. Since $W^1$ and $W^2$ are independent and $\tau$ is $\F^1_{\infty}$-measurable, the Brownian motion $W^2$ is independent of $\G^1_{\infty}=\F^1_{\infty}$. Noting that $\GG=\GG^1\vee\FF^2$, this implies that $W^2$ remains a Brownian motion in the filtration $\GG$ and, hence, the process $S$ given by \eqref{SDE-counter} is also a $\GG$-local martingale. Due to Theorem \ref{char-arb}, this implies that NA1, NA and NFLVR all hold in the enlarged financial market $\M^{\GG}$ on the global time horizon $[0,\infty]$. 

In the context of the present example, it is easy to show that the replication arguments used in the proofs of part \emph{(ii)} of Theorem \ref{attau}, Proposition \ref{arb2} and part \emph{(ii)} of Proposition \ref{arb-tau+eps} break down. Indeed, if $\tau$ avoids all $\FF^1$-stopping times, Lemma \ref{NikYor} applied to the filtration $\FF^1$ gives the existence of an $\FF^1$-local martingale $N=(N_t)_{t\geq0}$ with $N_0=1$ and $\lim_{t\rightarrow\infty}N_t=0$ $P$-a.s. such that $P(\tau>t|\F^1_t)=N_t/N^*_t$ for all $t\geq 0$. By the predictable representation property in the Brownian filtration $\FF^1$, there exists an $\FF^1$-predictable process $\varphi=(\varphi_t)_{t\geq0}\in L^2_{\text{loc}}(W^1)$ such that $N=1+\varphi\cdot W^1$. 
Moreover, due to the independence of $W^1$ and $W^2$, it is also easy to check that $P(\tau=\sigma)=0$ for every $\FF$-stopping time $\sigma$ (Assumption \ref{ass-avoid}) and $Z_t=N_t/N^*_t$, for all $t\geq 0$, where $N$ is also an $\FF$-local martingale. However, we cannot replicate the $\FF$-local martingale $N$ by trading in the risky asset $S$, since the two processes $N$ and $S$ are driven by the independent Brownian motions $W^1$ and $W^2$, respectively.
\end{Ex}

\section{Conclusions and extensions}	\label{S7}

In the present paper, we have dealt with the question of whether the additional information associated to an honest time does give rise to arbitrage. Under Assumptions \ref{ass-NFLVR}-\ref{ass-PRP}, we have given a complete and precise answer in the context of a general continuous financial market model. In particular, we have studied the validity of no-arbitrage-type conditions which go beyond the classical NFLVR criterion, such as the NUIP and NA1/NUPBR conditions. We have shown in a simple and direct way that an informed agent can realise arbitrage opportunities \emph{at} an honest time as well as \emph{after} an honest time, while arbitrages of the first kind can only be obtained by trading \emph{as soon as} an honest time occurs. On the other hand, it is impossible to make arbitrage profits strictly \emph{before} an honest time. The present paper significantly extends previous results in the literature, providing at the same time simpler and more transparent proofs.

We want to conclude by commenting on the role of Assumptions \ref{ass-NFLVR}-\ref{ass-PRP} and discussing some possible extensions and generalisations. 
The present paper aims at understanding the impact of an honest time on the validity of suitable no-arbitrage-type conditions in the \emph{enlarged} financial market $\M^{\GG}$ and, hence, we assumed from the beginning that the \emph{restricted} financial market $\M^{\FF}$ is free from any kind of arbitrage, in the classical sense of NFLVR (Assumption \ref{ass-NFLVR}). However, we want to point out that analogous results can be obtained if the restricted financial market $\M^{\FF}$ satisfies NA1 (or, equivalently, NUPBR) but the stronger NFLVR condition does not necessarily hold. In that case, Theorem \ref{attau} and Proposition \ref{arb2} continue to hold, provided that Assumptions \ref{ass-avoid}-\ref{ass-PRP} are still satisfied. Indeed, due to Remark \ref{rem-filtr}, if NFLVR fails to hold in the restricted financial market $\M^{\FF}$, then it fails in the enlarged financial market $\M^{\GG}$ as well (and can be also shown to fail on the time horizon $\left[0,\tau\right]$). Moreover, by relying on part \emph{(i)} of Theorem \ref{char-arb} together with Lemma \ref{dec} and Assumptions \ref{ass-avoid}-\ref{ass-PRP}, one can show that NA1 (or, equivalently, NUPBR) holds in the enlarged financial market $\M^{\GG}$ on the time horizon $\left[0,\tau\right]$ but fails on the global time horizon $\left[0,\infty\right]$. For the sake of brevity, we omit the details and refer the interested reader to Section 4.4.3 of \citet{Fon}.

The assumption that the honest time $\tau$ avoids all $\FF$-stopping times (Assumption \ref{ass-avoid}) seems to be crucial. Indeed, if NFLVR holds in the restricted financial market $\M^{\FF}$ (Assumption \ref{ass-NFLVR}) but Assumption \ref{ass-avoid} does not hold, then an honest time $\tau$ does not necessarily give rise to arbitrage opportunities in the enlarged financial market $\M^{\GG}$ on the time horizon $\left[0,\tau\right]$. As an example, let $\tilde{\tau}$ be an honest time which avoids all $\FF$-stopping times and let $\sigma$ be any $\FF$-stopping time such that $\tilde{Z}_{\sigma}>0$ $P$-a.s., where $\tilde{Z}=(\tilde{Z}_t)_{t\geq 0}$ is the Azéma supermartingale of $\tilde{\tau}$. Then $\tau:=\tilde{\tau}\wedge\sigma$ is easily seen to be an honest time (which does not avoid $\FF$-stopping times) and, as can be deduced from Theorem \ref{beforetau}, NFLVR still holds in the enlarged financial market $\M^{\GG}$ on $\left[0,\tau\right]$.

Observe that our results have been obtained under Assumption \ref{ass-PRP}, which implies that, under any ELMM$_{\FF}$ $Q$, the $\left(Q,\FF\right)$-local martingale $N$ appearing in the multiplicative decomposition of the Azéma $Q$-supermartingale of $\tau$ (see Lemma \ref{NikYor}) can be written as $N=1+\varphi\cdot S$. As can be easily checked (see in particular the proofs of part \emph{(ii)} of Theorem \ref{attau} and Proposition \ref{arb2} and Example \ref{rem-counterexample}), only the latter condition is necessary and, hence, the assumption that \emph{all} $\FF$-local martingales can be represented as stochastic integrals of $S$ can be significantly relaxed, provided that all $\FF$-local martingales are continuous.

Finally, we want to remark that the present paper gives a complete picture of the relations between honest times and arbitrage in the context of general financial market models based on \emph{continuous} semimartingales. However, at least under suitable additional assumptions, our results can be extended to the case where the discounted price process $S$ has possibly discontinuous paths. For instance, all the results of the present paper still hold if $S$ is assumed to be (under some ELMM$_{\FF}$ $Q$) a Lévy process and $\tau$ an honest time such that its Azéma supermartingale $Z$ admits a multiplicative decomposition as in Lemma \ref{NikYor}, such that $N=1+\varphi\cdot S=1+\varphi\cdot S^c$, where $S^c$ denotes the continuous $\left(Q,\FF\right)$-local martingale part of $S$. Note also that a result analogous to Lemma \ref{NikYor}, which plays a key role in the present paper, has been recently established in the discontinuous case by \citet{Kar2} (see also \citet{NP2}, Theorem 3.2). For reasons of space, we omit the details and postpone a complete study of the discontinuous case to a forthcoming work.
We also mention that, in a general semimartingale model, sufficient conditions for the stability of NA1 with respect to an arbitrary filtration expansion has been recently established in \citet{Shiqi}.

\vspace{1cm}
\begin{small}
\noindent\textbf{Acknowledgements:}
The authors are thankful to the Fédération Bancaire Française (within the chaire ``Risque de Crédit'' program) for generous financial support and to Anna Aksamit, Ashkan Nikeghbali, Marek  Rutkowski, Marc Yor, two anonymous referees and an associate editor for valuable comments that helped to improve the contents as well as the presentation of the paper.
\end{small}

\appendix
\section{Appendix}
\subsubsection*{Proof of Lemma \ref{comp}.}

Due to the representation \eqref{str-defl-sigma}, for any $\FF$-stopping time $\sigma$ and for any local martingale deflator $L$ in $\GG$ on the time horizon $[0,\sigma\wedge\tau]$, we can write:
\be	\label{comp-1}	\ba
E\left[L_{\sigma\wedge\tau}\right]
&= E\left[L_{\sigma}\ind_{\left\{\sigma<\tau\right\}}\right]
+E\left[L_{\tau}\ind_{\left\{\tau\leq\sigma\right\}}\right]	\\
&= E\left[\frac{1}{N_{\sigma}}\exp\left(-\int_0^{\sigma}\!\frac{k_s}{N^*_s}\,dN^*_s\right)
\ind_{\left\{\sigma<\tau\right\}}\right]
+E\left[\frac{1}{N_{\tau}}\exp\left(-\int_0^{\tau}\!\frac{k_s}{N^*_s}\,dN^*_s\right)\left(1+k_{\tau}+\eta\right)\ind_{\left\{\tau\leq\sigma\right\}}\right]\,.
\ea	\ee
Let us first focus on the first term on the right-hand side of \eqref{comp-1}. Recall that $\tau<\nu$ $P$-a.s. (see the proof of Proposition \ref{exdefl}) and that $Z=\left(Z_t\right)_{t\geq 0}$ is the $\FF$-optional projection of $(\ind_{\left\{\tau>t\right\}})_{t\geq 0}$ and $Z_{\sigma}/N_{\sigma}=1/N^*_{\sigma}$ on the set $\left\{\sigma<\nu\right\}$ (see Lemma \ref{NikYor}). Then, we can write as follows:
\be	\label{comp-2}	\ba
&E\left[\frac{1}{N_{\sigma}}\exp\left(-\int_0^{\sigma}\!\frac{k_s}{N^*_s}\,dN^*_s\right)
\ind_{\left\{\sigma<\tau\right\}}\right]
= E\left[\frac{1}{N_{\sigma}}\exp\left(-\int_0^{\sigma}\!\frac{k_s}{N^*_s}\,dN^*_s\right)
\ind_{\left\{\sigma<\tau\right\}}\ind_{\left\{\sigma<\nu\right\}}\right]	\\
&\mbox{} = E\left[\frac{1}{N_{\sigma}}\exp\left(-\int_0^{\sigma}\!\frac{k_s}{N^*_s}\,dN^*_s\right)
Z_{\sigma}\ind_{\left\{\sigma<\nu\right\}}\right]
= E\left[\frac{1}{N^*_{\sigma}}\exp\left(-\int_0^{\sigma}\!\frac{k_s}{N^*_s}\,dN^*_s\right)
\ind_{\left\{\sigma<\nu\right\}}\right]	\\
&\mbox{} = E\left[\exp\left(-\int_0^{\sigma}\!\frac{1+k_s}{N^*_s}\,dN^*_s\right)
\ind_{\left\{\sigma<\nu\right\}}\right]\,.
\ea	\ee
Let us compute more explicitly the second term on the right-hand side of \eqref{comp-1}. Recall that $E\left[\eta|\G_{\tau-}\right]=0$ (see \citet{JS}, Theorem 6.2). Recall also that, since all $\FF$-local martingales are continuous (due to Assumption \ref{ass-PRP} together with the continuity of $S$), Corollary 2.4 of \citet{NY} (see also \citet{MY}, Exercise 1.8) implies that the dual $\FF$-predictable projection of the process $(\ind_{\left\{\tau\leq t\right\}})_{t\geq 0}$ is given by $\bigl(\log\left(N^*_t\right)\bigr)_{t\geq 0}$ (we refer the reader to Section 4.3 of \citet{Nik} or Section I.3b of \citet{JacShi} for the definition and the properties of dual predictable projections). Moreover, the measure $dN^*_s$ is supported by the set $\left\{s\geq 0:Z_s=1\right\}=\left\{s\geq 0:N_s=N^*_s\right\}$. Then, we can write as follows, where the first equality follows by first taking the $\G_{\tau-}$-conditional expectation, recalling that $\{\tau>\sigma\}\in\G_{\tau-}$ (see \citet{JacShi}, \textsection I.1.17) and that $k_{\tau}$ is $\G_{\tau-}$-measurable (see \citet{JacShi}, Proposition I.2.4):
\be	\label{comp-3}	\ba
&E\left[\frac{1}{N_{\tau}}\exp\left(-\int_0^{\tau}\!\frac{k_s}{N^*_s}\,dN^*_s\right)\left(1+k_{\tau}+\eta\right)\ind_{\left\{\tau\leq\sigma\right\}}\right]
= E\left[\frac{1}{N_{\tau}}\exp\left(-\int_0^{\tau}\!\frac{k_s}{N^*_s}\,dN^*_s\right)\left(1+k_{\tau}\right)
\ind_{\left\{\tau\leq\sigma\right\}}\right]	\\
&\mbox{} = E\left[\int_0^{\sigma}\!\frac{1}{N_s}\exp\left(-\int_0^s\!\frac{k_u}{N^*_u}\,dN^*_u\right)\left(1+k_s\right)\frac{1}{N^*_s}\,dN^*_s\right]	 = E\left[\int_0^{\sigma}\!\frac{1}{N^*_s}\exp\left(-\int_0^s\!\frac{k_u}{N^*_u}\,dN^*_u\right)\left(1+k_s\right)\frac{1}{N^*_s}\,dN^*_s\right]	\\
&\mbox{} = E\left[\int_0^{\sigma}\!\exp\left(-\int_0^s\!\frac{1+k_u}{N^*_u}\,dN^*_u\right)\left(1+k_s\right)\frac{1}{N^*_s}\,dN^*_s\right]
= E\left[1-\exp\left(-\int_0^{\sigma}\!\frac{1+k_s}{N^*_s}\,dN^*_s\right)\right]\,.
\ea	\ee
Equation \eqref{comp-a} then follows by combining \eqref{comp-2} and \eqref{comp-3}, using the fact that, since $\tau<\nu$ $P$-a.s., we have $\sigma>\tau$ on the set $\left\{\sigma\geq\nu\right\}$ and noting that the process $N^*$ is constant after $\tau$. 
In order to show that $\int_0^{\tau}\!\frac{1+k_s}{N^*_s}\,dN^*_s>0$ $P$-a.s., it suffices to note the following, where we use the fact that $1+k_{\tau}>0$ $P$-a.s. (see Remark \ref{rem-rep}):
$$
E\left[\frac{1+k_{\tau}}{N^*_{\tau}}\ind_{\left\{\int_0^{\tau}\!\!\frac{1+k_s}{N^*_s}dN^*_s=0\right\}}\right]
= E\left[\int_0^{\infty}\frac{1+k_s}{N^*_s}\ind_{\left\{\int_0^s\!\!\frac{1+k_u}{N^*_u}dN^*_u=0\right\}}\frac{1}{N^*_s}\,dN^*_s\right] = 0
$$
thus implying that $\int_0^{\tau}\!\frac{1+k_s}{N^*_s}dN^*_s>0$ $P$-a.s.
The last assertion follows from the fact that the non-negative $\GG$-local martingale $L^{\sigma\wedge\tau}$ is a uniformly integrable $\GG$-martingale if and only if $E[L^{\sigma\wedge\tau}_{\infty}]=E[L_{\sigma\wedge\tau}]=1$. Due to \eqref{comp-a}, the latter holds if and only if $P(\nu\leq\sigma)=0$.
\qed

\end{document}